\documentclass[onefignum,onetabnum]{siamonline190516}
\usepackage{draftwatermark}
\SetWatermarkText{Author Preprint}
\SetWatermarkAngle{0}
\SetWatermarkVerCenter{750pt}
\SetWatermarkScale{1}

\newcommand{\secref}[1]{Section \ref{sec:#1}}
\newcommand{\condref}[1]{Condition \ref{cond:#1}}
\newcommand{\appref}[1]{Appendix \ref{app:#1}}
\newcommand{\figref}[1]{Figure \ref{fig:#1}}
\newcommand{\tableref}[1]{Table \ref{tab:#1}}
\newcommand{\equref}[1]{(\ref{eq:#1})}
\newcommand{\algoref}[1]{Algorithm \ref{alg:#1}}
\newcommand{\figloc}[1]{\emph{(#1)}}

\usepackage{amsmath,amsfonts,mathtools,bm}
\usepackage{amssymb}
\usepackage{xcolor}
\usepackage{colortbl}
\usepackage{wrapfig}
\usepackage{dblfloatfix}
\usepackage{nowidow}
\usepackage{algorithm}
\usepackage{booktabs}
\usepackage[noend]{algpseudocode}
\usepackage{changepage}
\usepackage{array}
\usepackage{multirow}
\usepackage{appendix}
\usepackage[nocompress]{cite}

\newcommand{\Break}{\State \textbf{break}}
\newcommand{\Method}[2]{\State \textbf{method} \textsc{#1} \(\left( #2 \right)\):}
\algrenewcommand\algorithmicindent{1.3em}

\definecolor{clblue}{RGB}{222,235,247}

\DeclareMathOperator*{\argmin}{argmin}
\newcommand{\ESG}{E_{\operatorname{G}}}
\newcommand{\fSG}{f_{\operatorname{G}}}
\newcommand{\ESD}{E_{\operatorname{D}}}
\newcommand{\fSD}{f_{\operatorname{D}}}
\newcommand{\ETutte}{E_{\operatorname{T}}}
\newcommand{\EConformal}{E_{\operatorname{C}}}
\newcommand{\EARAP}{E_{\operatorname{A}}}
\newcommand{\fARAP}{f_{\operatorname{A}}}
\newcommand{\Ro}{\mathbb{R}}

\newcommand{\Rnm}[2]{\mathbb{R}^{#1 \times #2}}
\newcommand{\SOn}[1]{\operatorname{SO}(#1)}
\newcommand{\SPDn}[1]{\mathcal{S}^{#1}_{+}}

\newcommand{\bV}{\mathbf{V}}
\newcommand{\bW}{\mathbf{W}}
\newcommand{\bJ}{\mathbf{J}}
\newcommand{\bP}{\mathbf{P}}
\newcommand{\bU}{\mathbf{U}}
\newcommand{\bQ}{\mathbf{Q}}
\newcommand{\bLambda}{\mathbf{\Lambda}}

\newcommand{\norm}[1]{\left\lVert #1 \right\rVert}
\newcommand{\normu}[1]{\lVert #1 \rVert}
\DeclareMathOperator{\rot}{rot}
\DeclareMathOperator{\atant}{atan2}
\newcommand{\transp}{\top}

\allowdisplaybreaks[4]

\usepackage{enumitem}
\setlist[itemize]{leftmargin=*}
\setlist[enumerate]{leftmargin=*}

\begin{document}

\newsiamthm{condition}{Condition}
\renewcommand*{\thecondition}{\Alph{condition}}
\newsiamthm{claim}{Claim}
\newsiamthm{remark}{Remark}
\newsiamthm{fact}{Fact}

\title{A Splitting Scheme for Flip-Free Distortion Energies}

\author{Oded Stein\thanks{Massachusetts Institute of Technology, Cambridge, MA.}
    \and Jiajin Li\thanks{The Chinese University of Hong Kong, Hong Kong.}
    \and Justin Solomon\thanks{Massachusetts Institute of Technology, Cambridge, MA.}}

\headers{A Splitting Scheme for Flip-Free Distortion Energies}{O. Stein, J. Li, and J. Solomon}

\maketitle

\begin{abstract}
    We introduce a robust optimization method for flip-free distortion
    energies used, for example, in parametrization, deformation, and volume
    correspondence.
    This method can minimize a variety of distortion energies, such
    as the symmetric Dirichlet energy and our new symmetric gradient energy.
We identify and exploit the special structure of distortion energies
    to employ an operator splitting technique, leading us to propose a novel
    \textbf{A}lternating \textbf{D}irection \textbf{M}ethod of
    \textbf{M}ultipliers (ADMM) algorithm to deal with the non-convex,
    non-smooth nature of distortion energies.
The scheme results in an efficient method where the global step involves a single matrix multiplication and the local steps are closed-form per-triangle/per-tetrahedron expressions that are highly parallelizable.
    The resulting general-purpose optimization algorithm exhibits 
    robustness to flipped triangles and tetrahedra in initial data as well as
    during the optimization.
We establish the convergence of our proposed algorithm
    under certain conditions and demonstrate applications to parametrization, deformation, and volume correspondence.
\end{abstract}

\begin{keywords}
	computer graphics, optimization, nonconvex optimization, parametrization, ADMM
\end{keywords}

\begin{AMS}
	65K10, 90C26, 65D18, 68U05
\end{AMS}

\begin{figure}[t]
 \includegraphics[width=\textwidth]{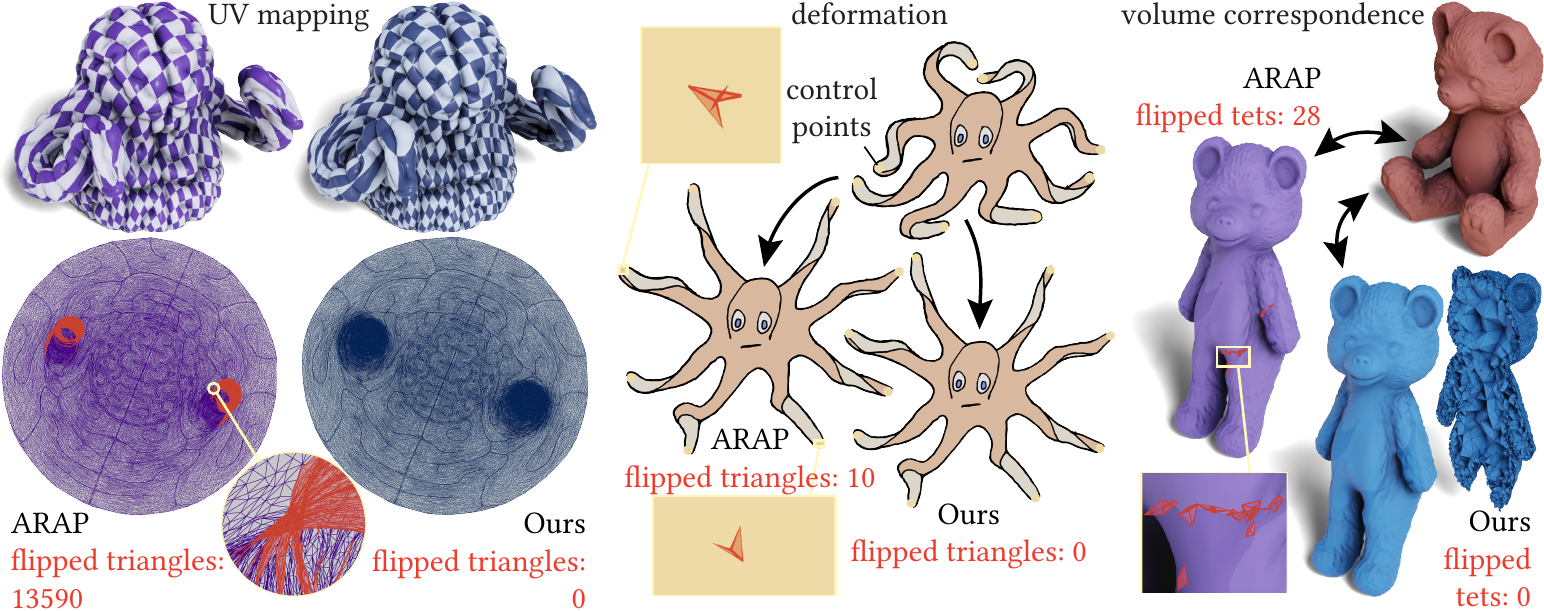}
 \caption{Minimizing distortion energies in a variety of applications using our
 splitting method:
 UV mapping \figloc{left, computing a distortion-minimizing map from the surface
 to \(\Ro^2\)},
 shape deformation \figloc{center, fixing control points to deformed position
 and find the distortion-minimizing map},
 volume correspondence \figloc{right, finding the distortion-minimizing map
 between the interior of two different surfaces}.
 Our method produces a flip-free result, unlike methods based on energies such
 as ARAP (\(\EARAP\)), which can exhibit flips when performing the same
 operation (flipped elements in red).
 \label{fig:teaser}}
\end{figure}

\section{Introduction}
\label{sec:section}

Distortion energies measure how much a mapping from one shape to another
deforms the initial shape.
Minimizing these energies can yield maps
between domains with as little distortion as possible.
Minimization of distortion energies with a variety of constraints is employed
in a wide array of computer graphics applications, such as UV mapping
(where one seeks to embed a 3D surface into 2D while minimizing distortion,
\figref{teaser} \emph{left}),
deformation (where parts of a surface or volume
are deformed, and the goal is to find the overall deformation with least
distortion, \figref{teaser} \emph{center}), and volume correspondence
(two boundary surfaces are given, and a distortion-minimizing map between the
two volumes is desired, \figref{teaser} \emph{right}).
We are interested in computing maps that minimize distortion energies on
triangle and tetrahedral meshes.

Flip-free distortion energies comprise an important subset of distortion
energies.
A mapping that minimizes such an energy will never invert (flip) a
triangle or tetrahedron.
Flip-free distortion energies are difficult to optimize:
they are usually non-linear and non-convex, and have singularities
that correspond to collapsed elements.
Typical optimization methods based on line-search
require feasible, flip-free iterates.
Thus, they must exercise great care to avoid singularities, 
where they will fail.
For applications such as volume correspondence, it is difficult to
even initialize with a feasible flip-free configuration.
Certain line-search-free approaches can struggle with the problem's
non-convexity in both objective function and feasible set.

We focus on distortion energies that depend only on the mapping's Jacobian,
are invariant to rotations, and are convex over symmetric positive definite
matrices.
This includes popular flip-free distortion energies such as the
\emph{symmetric Dirichlet energy}, as well as our new
\emph{symmetric gradient energy}.
Previous methods optimizing the symmetric Dirichlet energy can, to our
knowledge, not be mathematically proven to
converge
in the limit.
We exploit the convexity in these energies
by splitting the Jacobian of the mapping \(W\) into a rotational part \(U\)
and a flip-free, symmetric part \(P\). 
We propose a novel ADMM algorithm to leverage this splitting in an efficient
fashion, which results in three sub-problems: optimize the vertex positions of
the mapping \(W\), optimize the Jacobian's rotational part \(U\), and optimize
the Jacobian's rotation-free part \(P\). 
The optimization in \(W\) is linear,
the optimization in \(U\) is an explicitly-solvable Procrustes problem, and
the optimization in \(P\) (the only part containing the objective function) has a closed-form solution for both energies
considered in this article.
The optimizations in \(U\) and \(P\) also decouple over triangles/tetrahedra,
and can thus be parallelized, while 
the optimizations in \(W\) and \(P\) are convex.

In our approach, the target mapping does not need to be flip-free for every
iteration until convergence is attained---since \(P\) is always flip-free,
the distortion energy will never be singular
and can be evaluated even if \(W\) contains flips.
Thus, our approach is naturally robust to flipped elements in the
iterates and can be initialized with flipped elements.

ADMM-based algorithms are, in general, not guaranteed to converge for
non-convex nonlinear problems, such as the ones involving distortion energies.
For our method, however, we can present theoretical analysis of convergence
behavior that can show convergence given certain conditions.
This mathematical proof goes beyond what is usual for other flip-free
optimization algorithms, yielding the first optimization of the symmetric
Dirichlet energy that can be proven to converge.
Beyond describing the circumstances under which we reach a critical point of
the optimization problem, this analysis significantly informs our algorithm:
it lets us automatically choose appropriate augmented Lagrangian penalty
weights.
\\

Our contributions are:
\begin{itemize}
    \item a parallelizable optimization method for non-linear non-convex
    flip-free distortion energies that is robust to the presence
    of flipped triangles in the initial data;
    \item a convergence analysis that discusses the convergence of our algorithm
    to a stationary point given certain conditions; 
    \item a novel distortion energy, the \emph{symmetric gradient energy},
    which yields flip-free distortion-minimizing maps that are similar to
    popular non-flip-free methods (but are flip-free).
\end{itemize}
We demonstrate our method on applications in UV parametrization, surface and
volume deformation, and volume correspondence (see \figref{teaser}).

\section{Related Work}
\label{sec:relatedwork}

\subsection{Optimizing Distortion Energies}
\label{sec:relatedsurfaceparametrization}

Distortion energies have a long history in geometry processing and
related fields like physical simulation and differential geometry.
They are part of tools for parametrization, deformation, and related tasks.

Early approaches include optimizing harmonic and conformal energies to produce
angle-preserving mappings with a variety of optimization methods
\cite{Eck1995,Desbrun2002,Levy2002,Gu2003,Sheffer2004,Mullen2008}.
As they can be measured and optimized efficiently, conformal energies
remain popular and are the subject of ongoing research
\cite{Springborn2008,Sawhney2017,Sharp2018,Soliman2018,Gillespie2021}.

A different way to measure distortion is to quantify the deviation of a
mapping's local structure from rotation, as in the
As-Rigid-As-Possible (ARAP) energy \cite{Sorkine2007}, which can be efficiently
optimized using a per-element local-global approach \cite{Liu2008}.
ARAP is similar to other quasi-elastic energies \cite{Chao2010}.

Optimization of flip-free distortion energies goes back to the work of Tutte
\cite{Tutte1963}, who showed that minimizing Tutte's energy while fixing
boundary vertices to a convex polygon yields a flip-free mesh parametrization.
Recent methods compute flip-free maps while simultaneously minimizing
some kind of distortion for surfaces \cite{Lipman2012,Kovalsky2015},
volumes \cite{Aigerman2013}, simplicial maps \cite{Lipman2014},
non-standard boundary conditions \cite{Weber2014}, or with a focus on
numerical robustness \cite{Shen2019}.
Conformal and harmonic energies can be augmented to produce flip-free
maps, e.g.\ by including cone singularities in the parametrization 
\cite{Hefetz2019}.
Cone singularities can be used in a variety of ways to produce parametrizations
\cite{Soliman2018,Chien2016}.

The symmetric Dirichlet energy \cite{Schreiner2004,Smith2015} combines the idea
of measuring the deviation of the Jacobian from the identity with flip-free
maps:
the energy is singular for zero-determinant Jacobians, which means that during
the minimization process elements can not collapse and invert.
Since the symmetric Dirichlet energy is non-convex and singular, it
requires specialized optimization algorithms.
The symmetric Dirichlet energy can be optimized in a wide variety of ways:
with a modified line-search to avoid the energy's
singularities in a L-BFGS-style optimization (augmented with techniques for
global bijectivity) \cite{Smith2015},
with a quadratic proxy to accelerate convergence \cite{Kovalsky2015},
with a local-global modification of line-search that can also be applied to
a variety of other rotation-independent distortion energies
\cite{Rabinovich2017},
with a preconditioned line search based on Killing field approximation \cite{Claici2017},
and by progressively adjusting the reference mesh \cite{Liu2018}.
These approaches require initialization with a flip-free map whose distortion is
then further reduced.

A different approach to generating flip-free distortion-minimizing relies on
custom distortion
energies that can be optimized efficiently without
line-search methods \cite{Garanzha2021}.
Concurrent work introduces a framework for flip-free conformal maps
\cite{Gillespie2021}, and a framework for globally elastic 3D deformations
using physical simulation methods \cite{Fang2021}.
Yet other approaches include initializing with a flip-free map by lifting the
mesh to a higher dimension and achieving injectivity that way \cite{Du2020},
and applying block coordinate descent \cite{Naitsat2018}.
Using a barrier-aware line-search and quasi-Newton methods, one can optimize
a variety of distortion energies \cite{Zhu2018}.
One can also discretize physical elasticity energies that naturally model
distortion and carry out a physical simulation \cite{Smith2019}.

The concurrent work WRAPD \cite{Brown2021} also uses ADMM to compute flip-free
distortion-minimizing maps, but with a different splitting technique than ours. 
They use two-block ADMM (compared to our three blocks with
variables \(\bW\), \(\bU\), \(\bP\)), where the non-convex ADMM sub-step
requires a line-search optimization to be solved to convergence every step
(without guarantees on how this might interfere with the ADMM's convergence).
They do not include a convergence proof (which we do).

There are many other approaches for efficient flip-free distortion minimization
\cite{Aigerman2014,Fu2015,Naitsat2018,Su2019,Su2020,Choi2018,Yueh2019}
For certain applications, such as quad meshing \cite{Kalberer2007},
surface-to-surface mapping \cite{Ezuz2019,Schmidt2019,Schmidt2020},
joint optimization of map and domain \cite{Li2018}, globally
bijective mapping \cite{Jiang2017}, and input-aligned maps \cite{Myles2014},
distortion energies with special properties are used.

\subsection{Alternating Direction Method of Multipliers}
\label{sec:relatedaugmentedlagrangian}

The augmented Lagrangian method and the alternating direction method of
multipliers (ADMM) are popular optimization methods \cite{Boyd2011}.
While the convergence of ADMM is well-known for convex problems,
convergence can also be proven in some other scenarios that often necessitate
special proofs, such as classes of weakly convex problems
\cite{zhang2019fundamental}, certain non-convex ADMMs with \emph{linear}
constraints \cite{Zhang2020,Wang2019,Hong2016}, non-convex and non-linear, but
equality-constrained problems \cite{Wang2021}, and bilinear constraints
\cite{Zhang2021}.
For specific non-linear non-convex ADMMs, specialized proofs exist
\cite{gao2020admm,Yang2017}.
Our method does not exactly fit any of the above approaches, but uses ideas
from many of them to analyze convergence, such as an explicit boundedness
condition \cite{Zhang2021}, the use of a potential function
\cite{Zhang2020},
and the K\L{} condition \cite{gao2020admm}.

ADMM has been employed in many computer graphics and image
processing applications.
It is especially useful when a convex problem can be split into multiple
simpler sub-problems that are each convex -- in that case, convergence of the method
follows from standard results \cite[Section 3.2]{Boyd2011}.
Such convex ADMM is used, for example, to produce developable surfaces after
convex relaxation \cite{Sellan2020}, to speed up optimization in computer vision
and machine learning \cite{Xu2017}, for aligning point sets with rigid
transforms through relaxation of a non-convex problem \cite{Sanyal2017},
as a sub-step in a rotation-strain simulation of elastic objects \cite{Pan2015},
to compute the Earth Mover's Distance \cite{Solomon2014},
and for isogeometric analysis after transforming a non-convex problem into a
biconvex problem \cite{Nian2016}.

When a problem is non-convex, ADMM is more difficult to employ in a way that
assures convergence.
As a result, some applications of ADMM do not provide an explicit convergence
guarantee but are able to show convergence empirically.
Past work uses non-convex ADMM with a linear constraint for the physical
simulation of elastic bodies with collisions \cite{Overby2017}, an approach
applied later to character deformation \cite{Minor2018} and cloth simulation
\cite{Minor2018cloth}.
In later concurrent work this approach is extended to globally injective maps
\cite{Overby2021} by adding a step to the ADMM that promotes
injectivity through a nonlinear optimization procedure
while temporarily tolerating non-injective maps;
this ADMM uses a line-search in the inner loop.

For some specific non-convex applications of ADMM, convergence can be proven,
just as we do in this article,
although these examples do not cover our use case
\cite{zhang2019accelerating,Ouyang2020}.
\\

Unlike past applications of ADMM to the problem of distortion-minimizing maps,
our distortion minimization technique has all of the following features:
\begin{itemize}
    \item We solve a non-convex problem with a \emph{non-linear constraint},
    \((G\bW)_i = \bU_i \bP_i\).
    \item Our splitting contains \emph{three} ADMM blocks instead of the usual
    two, designed so that each block is solvable in closed-form.
    \item We present a convergence analysis specialized to our
    algorithm; to our knowledge, this theoretical analysis is new and adds
    to the cases in which non-convex multi-block ADMM is proven to converge.
\end{itemize}

\section{Problem setup}

\subsection{Preliminaries}
\label{sec:preliminaries}

We compute a map from a \emph{source mesh} to a deformed
\emph{target mesh} composed of triangle or tetrahedra.
The goal is to measure the distortion of the map and to optimize the map
(and the target mesh) subject to certain constraints.
\(\bV \in \Rnm{n}{d_\iota}\) contains the coordinates of the vertices of the
source mesh, where \(n\) is the number of vertices and \(d_\iota=2,3\) is
their dimension.
The individual vertices are denoted by \(\bV_i, i=1,\dots,n\).

The optimization variable \(\bW \in \Rnm{n}{d_o}\) contains the target
coordinates of the vertices under the map, where
\(d_o \leq d_i\)
is the dimension of the output vertices. 
The dimension \(d_o\) 
can be different from the input dimension, for example when we compute the
UV mapping of a surface in 3D, where \(d_\iota=3\) and \(d_o=2\).
The individual vertices are denoted by \(\bW_i, i=1,\dots,n\).
The number of triangles or tetrahedra (elements) in the
mesh is \(m\) and the dimension of the elements is \(d\)
(\(d=2\) for triangles and \(d=3\) for tetrahedra).
\(w_i\) is the area or volume of the \(i\)-th element, depending on \(d\).
We will deal with collections of matrices associated with each
triangle or tetrahedron of a mesh.
For this purpose, we
    let \((\Rnm{d}{d})^m\) be the space of \(m\) independent \(d \times d\)
    matrices.
    If \(\bJ \in (\Rnm{d}{d})^m\), we denote by \(\bJ_i\) the \(i\)-th matrix in
    \(\bJ\).

Our energies depend on the Jacobian of the map \(\bV\mapsto\bW\).
Assuming our map is affine when restricted to the interior of each element,
the Jacobian is piecewise constant.
\begin{definition}[Piecewise constant Jacobian]
    \label{def:jacobian}
    Let \(\bV \in \Rnm{n}{d_\iota}\) be the source coordinates and
    \(\bW \in \Rnm{n}{d_o}\) the target coordinates of a mapping of a mesh with
    \(m\) triangles or tetrahedra. 
Then \(G : \Rnm{n}{d_o} \rightarrow (\Rnm{d_o}{d_o})^m\)
    is the linear operator (dependent on \(\bV\))
    such that the mapping's Jacobian for the \(i\)-th triangle
        or tetrahedron
    is given by the matrix \((G\bW)_i\).
\end{definition}
Supplemental material describes how to compute the Jacobian from target
vertex positions.
\smallskip

We use the Frobenius product and norm for vectors and matrices.
    The \emph{Frobenius product}
    (or \emph{dot product} for vectors) is defined as
    \(
        X \cdot Y \coloneqq \operatorname{trace} X^\transp Y
        \;\textrm{.}    
        \)
    The \emph{Frobenius norm} (or \(L^2\) norm for vectors) is defined via
\(
        \norm{X}^2 \coloneqq X \cdot X
        \textrm{.}
\)

At last we define spaces for the rotation and flip-free parts of the Jacobian.
\begin{definition}[Rotation and semidefinite matrices]
   \(\SOn{d}\) is the space of rotation matrices of dimension \(d\). 
   \(\SPDn{d}\) is the space of symmetric positive definite (spd) matrices of
   dimension \(d\).
\end{definition}

\subsection{Optimization Target}

We can understand a variety of algorithms for parametrization, deformation, and
related tasks as optimizing a generic energy of the following form:
\begin{equation}\label{eq:genericdeformationenergywithflips}
    E_{\mathrm{generic}}(\bW) \coloneqq \sum_{i=1}^m
    w_i f\big((G\bW)_i\big) 
    \;\textrm{.}
\end{equation}
Here, \(f(\cdot)\) denotes a per-element distortion energy, summed over the
triangles/tetrahedra \(i\) of the mesh;
we call \(f\) the \emph{defining function} of our problem.
It is evaluated on the Jacobian \(G\bW\) of the map \(\bV\mapsto\bW\) and is
typically designed to be a rotation-invariant function quantifying how much
\((G\bW)_i\) deviates from being a rigid motion.
Rotation invariance implies we can think of \(f(\cdot)\) as a function of the
singular values of \((G\bW)_i\).

Each possible choice of distortion energy \(f(\cdot)\) captures a different
trade-off between different means of deforming the source domain onto the
target, e.g., between angle and area preservation.
Below, we discuss some choices for \(f(\cdot)\) used in our experiments;
we refer the reader to \cite[Table I]{Rabinovich2017} for an exhaustive list,
many of which can be easily plugged into our optimization framework.
In \secref{symmetricgradient}, we propose one additional option, the symmetric
gradient energy, which appears to yield flip-free maps in practice that are
similar to the popular---but not flip-free---ARAP energy (discussed
in \appref{arapse}).

\begin{figure}
    \includegraphics[width=\textwidth]{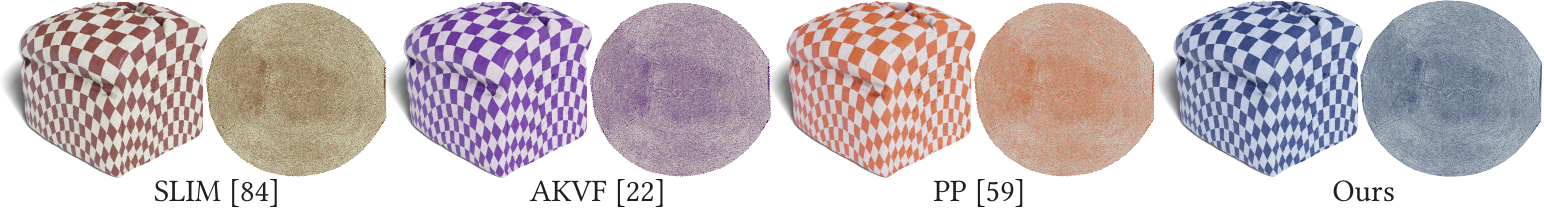}
    \caption{
        Our optimization method visually matches the results of a variety
        of other methods when tasked with producing a UV map with the same energy
        \(\ESD\).
    \label{fig:symmetricdirichletinoursandprevious}}
\end{figure}

Of particular interest are distortion energies that explicitly avoid inverted
elements in the computed mapping.
To achieve this, we augment \eqref{eq:genericdeformationenergywithflips} with a
term that forbids flipped triangles or tetrahedra, and add conditions on \(f\):
\begin{equation}\label{eq:genericdeformationenergy}
    \boxed{
    E(\bW) \coloneqq \sum_{i=1}^m
    w_i f\big((G\bW)_i\big) + \chi_{+}\big(\det (G\bW)_i\big)
    \;\textrm{.}
    }
\end{equation}
Here, \(\chi_+\) denotes the indicator function
\[
\chi_+(x)\coloneqq
\begin{cases}
\;0 & \textrm{ if }x \geq 0\\
\;\infty & \textrm{ otherwise.}
\end{cases}
\]
    If \(f\) is smooth on all Jacobians with positive determinant, then the energy
    \(E\) is smooth on all maps with positive determinant.
The extra term in \equref{genericdeformationenergy}
containing the characteristic function \(\chi_+\)
can be understood as a constraint preserving local injectivity of the map in the
interior of each element.
This needs to be combined with an appropriate \(f(x)\) which goes to \(\infty\)
for \(x \rightarrow 0\), to create an appropriate barrier that ensures the
region where \(\chi_+\) is infinite is never reached.
This constraint is implicit in the design of many past algorithms
\cite{Smith2015,Rabinovich2017,Claici2017,Liu2018}, but we choose
to expose it explicitly as it will inform our design in
\secref{ouraugmentedlagrangianmethod}.

The main thrust of our research is to propose an efficient algorithm for
optimizing energies of the form \equref{genericdeformationenergy}.
\tableref{allenergies} offers a quick overview over all energies that are
used in this article.
\figref{symmetricdirichletinoursandprevious} shows our optimization
reproducing a parametrization with an energy popular in previous work, the
symmetric Dirichlet energy \(\ESD\).

\begin{table}
\centering
    \begin{tabular}{p{52pt}|p{36pt}|p{120pt}|p{68pt}|p{85pt}}
        Energy & Symbol & Definition & Discussed in... & Properties \\
        \hline

        \rowcolor{clblue}
        Tutte & \(\ETutte\) &
        {\scriptsize
            \(\ETutte(\bW) = \sum\limits_{\textrm{edges }(i,j)}
            \frac{\norm{\bW_i - \bW_j}^2}{\norm{\bV_i - \bV_j}} \)
        } &
        \appref{tutte} &
        {\scriptsize
            flip-free in 2D (not 3D); linear; no initialization needed
        } \\

        Conformal & \(\EConformal\) &
        {\scriptsize
            \equref{genericdeformationenergywithflips} with
            \(f(X) = \frac12 \norm{X}^2\) plus additional area term
        } &
        \appref{conformal} &
        {\scriptsize
            not flip-free; linear; no initialization needed
        } \\

        \rowcolor{clblue}
        ARAP & \(\EARAP\) &
        {\scriptsize
            \equref{genericdeformationenergywithflips} with
            \(f(X) = \fARAP(X) = \frac12 \norm{X - \rot X}^2\), 
            \(\;\;\rot X\) is the rotational part of \(X\)
        } &
        \appref{arapse} &
        {\scriptsize
            not flip-free; nonlinear; desirable rigidity property; needs
            initializer; efficient optimizer
} \\

        Symmetric Dirichlet & \(\ESD\) &
        {\scriptsize
            \equref{genericdeformationenergy} with
            \(f(X) = \fSD(X) = \frac12 \left(
            \norm{X}^2 + \norm{X^{-1}}^2 \right)\)
        } &
        \appref{symmetricdirichlet} &
        {\scriptsize
            flip-free; strong singularity; nonlinear;
            needs initializer
        } \\

        \rowcolor{clblue}
        Symmetric gradient & \(\ESG\) &
        {\scriptsize
            \equref{genericdeformationenergy} with
            \(f(X) = \fSG(X) = \frac12 \norm{X}^2 - \log\det X\)
        } &
        \secref{symmetricgradient} &
        {\scriptsize
            flip-free; weak singularity; nonlinear;
            needs initializer
        }
    \end{tabular}
    \quad\\ \quad\\
    \caption{A table summarizing all deformation energies featured in this
    article.
    Energies from previous work are introduced in more detail in
    \appref{deformationenergies};
    the symmetric gradient energy is introduced in \secref{symmetricgradient}.
    \label{tab:allenergies}}
\end{table}

\section{Our Optimization Method}
\label{sec:ouraugmentedlagrangianmethod}

Our optimization method can be applied to distortion energies of the form
\equref{genericdeformationenergy}
where the defining function \(f\) is convex over the set of symmetric positive
semidefinite matrices \(\SPDn{d}\) (a property which holds for many distortion
energies).
There are two main challenges in optimizing energies
of this form.
One challenge is the non-convexity of the defining function \(f\) when
applied to arbitrary matrices \((G\bW)_i \in \Rnm{d}{d}\).
Another challenge is the non-smoothness of the characteristic function
\(\chi_{+}\) and potential singularities in \(f\).
Previous work solves these issues by, e.g., employing
line search methods which can handle non-convex problems,
and are specifically constructed to avoid the singular regions of
\(\chi_{+}\) and \(f\) during time stepping (see the discussion in
\secref{relatedwork}).

We address these challenges differently, starting
with the polar decomposition \cite[Theorem 2.17]{Hall2015}.
Every matrix \(J \in \Rnm{d}{d}\) with positive determinant can be decomposed
into a rotation matrix \(U \in \SOn{d}\) and a symmetric matrix
\(P \in \SPDn{d}\) such that
\begin{equation}\label{eq:polardecomposition}
J = U P
\;\textrm{.}
\end{equation}

Applying the decomposition to our problem, we can reformulate
\equref{genericdeformationenergy} as 
\begin{equation}\label{eq:admm-os}
\begin{split}
    & \min\limits_{\substack{\bW\\
        \bU \in (\SOn{d})^m\\
        \bP \in (\SPDn{d})^m}} \;
    \sum_{i=1}^m  w_i f(\bP_i)\\
    & \quad \quad \text{s.t. }  \quad \,(G\bW)_i\ - \bU_i \bP_i =0,
    \quad \forall i
    \;\textrm{.}
\end{split}
\end{equation}
This new formulation is now convex in both \(\bW\) and \(\bP\) if \(f\) is
convex over the positive semidefinite cone \(\SPDn{d}\):
\(\fSG\) and \(\fSD\) are non-convex for arbitrary matrices,
but convex for matrices in \(\SPDn{n}\), like many distortion energies.
The non-convexity is now entirely contained in \(\bU\).
We have thus extracted the \emph{hidden convexity of the problem}, and
restricted the non-convexity to rotational matrices only
(this extraction of salient parts of the energy mirrors approaches that
extract the singular values
\cite{Smith2015,Rabinovich2017,Claici2017,Liu2018} and appears in concurrent
work \cite{Brown2021}).

The constraint from \equref{admm-os} is still difficult to accommodate.
We deal with this problem by employing an ADMM \cite{Boyd2011} tailored
for our problem.
Let \(\bW\) be our target vertex positions, and let
\(\bU \in (\SOn{d})^m, \;\bP \in (\SPDn{d})^m, \;\bLambda \in (\Rnm{d}{d})^m\).
Our \emph{augmented Lagrangian function} is
 \begin{equation}\label{eq:augmentedlagrangian}
    \boxed{
    \!\!
    \begin{array}{r@{\ }l}
    \Phi(\bW, \bU, \bP, \bLambda) \coloneqq&
    \sum_{i=1}^m w_i f(\bP_i)  + \frac{\mu_i}{2} \left(\norm{(G\bW)_i - \bU_i \bP_i + \bLambda_i}^2
    - \norm{\bLambda_i}^2 \right) \!\!
    \end{array}
    }
\end{equation}
where \(\mu_i > 0\) are a set of \(m\) Lagrangian penalty weights.
The variable \(\bLambda\) is a scaled Lagrange multiplier for the constraint
\((G\bW)_i = \bU_i \bP_i\).

    In the formulation of \equref{augmentedlagrangian},
    the function \(f\) that contains a singularity for matrices with zero
    determinant is only ever evaluated on \(\bP_i\), which \emph{cannot}
    have zero determinant, as it is symmetric and positive definite.
    If the Jacobian map \((G\bW)_i\) inverts or degenerates a triangle, i.e.,
    has zero or negative determinant, this manifests as a
    feasibility error (which is finite).

\begin{figure}
    \includegraphics[width=\linewidth]{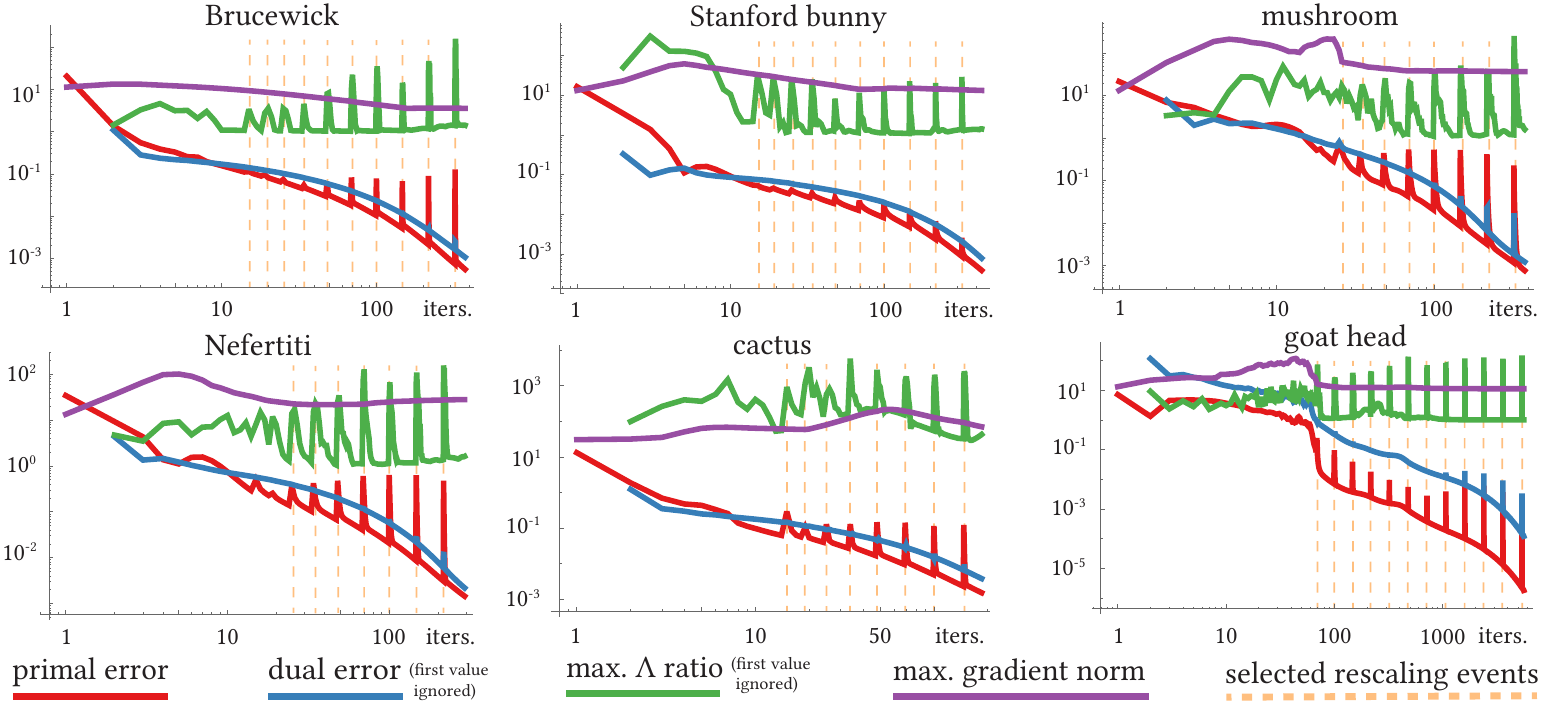}
    \caption{Log-log plots of the optimizations performed in \figref{uvmapping}
    showing \(e^{\textrm{prim}}\), \(e^{\textrm{dual}}\),
    the largest energy gradient \(\|\nabla f(\bP_i^{(k)})\|\) (as a proxy for
    \(B_i\) from \condref{gradf}, and the largest ratio
    \(\scriptstyle \norm{\Lambda_i^{(k+1)} - \Lambda_i^{(k)}} / \norm{\frac{1}{2} \left( \Lambda_i^{(k+1)} - \Lambda_i^{(k)} + U_i^{(k+1)}(\Lambda_i^{(k+1)})^\transp U_i^{(k+1)} - U_i^{(k)}(\Lambda_i^{(k)})^\transp U_i^{(k)} \right)} \)
    (as a proxy for \(\gamma^{1/2}\) from \condref{lambdabound},
    where the \(\Lambda_i\) are scaled by \(\mu_i\) to be able to compare
    them across rescalings).
    The errors steadily decrease, and gradient \& \(\Lambda\) ratio remain
    bounded, except for rescaling events that lead to temporary spikes that the
    bounds quickly recover from.
    Rescaling is not part of \algoref{mainalgorithm} for which we
    analyze convergence, but are employed in \algoref{actualalgorithm} to speed
    up the performance.
    \label{fig:uvmappingbounds}}
\end{figure}

The augmented Lagrangian method now consists of successively optimizing \(\Phi\)
in each of its primal arguments in an alternative way, and then updating the
dual variable \(\Lambda\).
Our method is described in pseudocode form in \algoref{mainalgorithm};
a description of each of the substeps follows.
\begin{algorithm}
\caption{Three-block ADMM for flip-free distortion energies
\label{alg:mainalgorithm}}
\begin{adjustwidth}{-0.5em}{}
\begin{algorithmic}[1]
    \Method{SplittingOptimization}{\bW^{(0)}, \bU^{(0)}, \bP^{(0)}, \bLambda^{(0)}}
    \For{\(k \gets 1, \dots\)}
        \State \(\!\! \bW^{(k)} \gets \argmin_{\bW\!, A\bW=b} \,
        \Phi\left(\bW, \bU^{(k-1)}, \bP^{(k-1)}, \bLambda^{(k-1)}\right) \)
        \State \(\bU^{(k)} \gets \argmin_{\bU} \Phi\left(\bW^{(k)}, \bU, \bP^{(k-1)}, \bLambda^{(k-1)}\right)
         + p\!\left(\bU,\bU^{(k-1)}\right) \) \State \(\, \bP^{(k)} \gets \argmin_{\bP} \,
        \Phi\left(\bW^{(k)}, \bU^{(k)}, \bP, \bLambda^{(k-1)}\right) \)
        \State \(\mkern 1mu \bLambda^{(k)}_i \gets \bLambda^{(k-1)}_i
        +\, ( G \bW^{(k)} )_i - \bU_i^{(k)} \bP_i^{(k)} \quad \forall i\)
    \EndFor
\end{algorithmic}
\end{adjustwidth}
\end{algorithm}

\subsubsection*{Updating \(\bW\)}
To update \(\bW\) we optimize \(\Phi\) with respect to \(\bW\).
\(\Phi\) is a quadratic function in \(\bW\), and can thus be optimized by
solving the linear system
\begin{equation}\begin{split}\label{eq:solvinginw}
    L \bW &= G^\transp r
    \;\textrm{,}
\end{split}\end{equation}
where
\begin{equation*}\begin{split}
r_i = \mu_i \left( \bU_i \bP_i - \Lambda_i \right)\in (\Rnm{d}{d})^m\qquad \textrm{and}\qquad
L = \sum_{i=1}^m \mu_i G_i^\transp G_i \in \Rnm{n}{n}
    \;\textrm{.}
\end{split}\end{equation*}
As \(G\) is implemented as a simple finite element gradient matrix (see
supplemental material), we can write \(L = G^\transp M G\), where \(M\) is a
mass matrix with the respective entries of \(\mu_i\) on the diagonal.
As a result, it is a sparse Laplacian matrix similar to the cotangent Laplacian
\cite{Pinkall1993}.

At this step, we can also enforce constraints, e.g., for deformation, on
the vertex positions \(\bW\).
The quadratic optimization problem can be solved, with any feasible linear
constraint of the form \(A\bW = b\), at negligible additional cost.
In fact, since \(G\) maps constant functions to \(0\), there needs to be a
minimum number
of
constraints to make \equref{solvinginw} solvable;
in the absence of any constraints (such as in UV mapping) we simply fix the
first vertex of \(\bW\) to the origin.

\subsubsection*{Updating \(\,\bU\)}
To update \(\bU\), we optimize \(\Phi\) augmented with a proximal
function \(p\),
\begin{equation}\label{eq:proximalfunctionu}
    p(\bU, \bU^{(k-1)}) \coloneqq
    \sum_{i=1}^m \frac{h_i}{2} \norm{\bU_i - \bU_i^{(k-1)}}^2
    \;\textrm{,}
\end{equation}
where \(h_i > 0\) is the proximal parameter and \(\bU^{(k-1)}\) is
the iterate from a previous step.
This proximal function \(p\) is needed for the algorithm to converge
(see \secref{convergence}).

Both \(\Phi\) and \(p\) decouple in \(\bU\) over elements, giving the problem
\begin{equation}\label{eq:rawoptimizationinudecoupled}
    \argmin_{\bU_i \in \SOn{d}} 
    \frac{\mu_i}{2} \norm{(G\bW)_i - \bU_i \bP_i + \bLambda_i}^2 +
    \frac{h_i}{2} \norm{\bU_i - \bU_i^{(k-1)}}^2
    \;\textrm{.}
\end{equation}
Since \(\bU_i \in \SOn{d}\) and \(\bP_i \in \SPDn{d}\),
we get that \equref{rawoptimizationinudecoupled} is equivalent to
the Procrustes problem \cite{Gower2004}
\begin{equation}\label{eq:optimizationinudecoupled}
    \argmin_{\bU_i \in \SOn{d}} 
    \norm{\bU_i - \bQ_i}^2
    \;\textrm{,}
\end{equation}
where
\(
    \bQ_i = \left( (G\bW)_i + \bLambda_i \right) \bP_i
    + \frac{h_i}{\mu_i} \bU_i^{(k-1)}
    \;\textrm{.}
\)
Supplemental material describes our approach to solving the Procrustes
problem in detail;
there is an explicit closed-form solution.

\subsubsection*{Updating \(\bP\)}
To update \(\bP\) we optimize \(\Phi\) with respect to \(\bP\).
As \(\Phi\) decouples in \(\bP\) over elements, we can solve optimize it
separately for each triangle/tetrahedron,
\begin{equation}\label{eq:optimizinginp}
    \argmin_{\bP_i \in \SPDn{d}}\; w_i f(\bP_i) +
    \frac{\mu_i}{2} \norm{(G\bW)_i - \bU_i \bP_i + \bLambda_i}^2
    \;\textrm{.}
\end{equation}
Since both \(\fSG\) and \(\fSD\) are convex over \(\SPDn{d}\), the problem
\equref{optimizinginp} is convex.
We merely need to find the single critical point by finding the solution in
\(\SPDn{d}\) of
\begin{equation}\label{eq:rootfindinginp}
    w_i \nabla f(\bP_i) + \mu_i \bP_i
    = \mu_i \operatorname{symm}\left( \bU_i^\transp \left( (G\bW)_i + \bLambda_i \right) \right)
    \;\textrm{,}
\end{equation}
where \(\operatorname{symm}(\cdot)\) symmetrizes a matrix: 
\(\operatorname{symm}(X) = \frac12 (X + X^\transp)\).
This can be solved explicitly in closed-form for both \(\fSG\) and \(\fSD\).
Due to floating point issues, closed-form solvers for \(\fSD\) can fail in
certain scenarios, we then use a simple iterative scheme (see supplemental
material).

\subsubsection*{Updating \(\bLambda\)}
To update the estimate of the Lagrange multiplier \(\Lambda\),
we apply a gradient ascent approach for the scaled augmented Lagrangian
method \cite[Section 3.1.1]{Boyd2011},
\begin{equation}\label{eq:lambdaupdate}
    \bLambda_i = \bLambda_i^{(k-1)} + (G\bW)_i - \bU_i \bP_i
    \;\textrm{,}
\end{equation}
where \(\bLambda^{(k-1)}\) is the iterate from a previous step of the
optimization method.
\\

As we will see in \secref{convergence}, \algoref{mainalgorithm} can be proven to
converge under certain conditions.
The algorithm requires the choice of both a Lagrangian penalty parameter
\(\mu_i\), as well as a proximal parameter \(h_i\).
The choice of both of these will be informed directly by the proof.
The actual algorithm implemented in our code (which is slightly different), the
initialization of \(\bW^{(0)}, \bU^{(0)}, \bP^{(0)}, \bLambda^{(0)}\), as well
as the termination condition are discussed in \secref{algorithm}.

\subsubsection*{Robustness to Flipped Elements}
\begin{figure}
    \includegraphics[width=\linewidth]{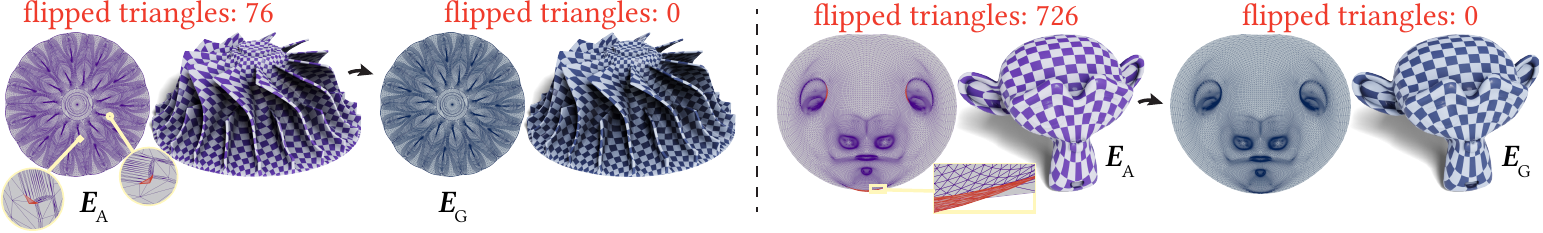}
    \caption{Using the robustness of our method with respect to flipped
    triangles in the initial iterate, we can use our method to unflip the
    outputs of other methods.
    In this example, a flip-containing UV parametrization computed with
    \(\EARAP\) is unflipped by running our method to optimize \(\ESG\) until
    a flip-free configuration is obtained, resulting in a parametrization that
    is very similar to \(\EARAP\)'s, but flip-free
    (flipped triangles in red).
    \label{fig:unflippingtriangles}}
\end{figure}
\algoref{mainalgorithm} can be robust with respect to flipped triangles in the
target mesh iterate \(\bW\).
Since the defining energy function \(f\) is only ever evaluated on the iterate
\(\bP\), which consists of
symmetric positive definite matrices (\(\SPDn{d}\))
this evaluation can never be undefined, \emph{even if} the target mesh iterate
\(\bW\) currently contains flipped triangles.
This is a property of the augmented Lagrangian method's weakly enforced
constraint.
\((G\bW)_i = \bU_i \bP_i\) is only ever strongly active when the algorithm has
converged, allowing us to circumvent the problem of evaluating \(f\) of a
singularity which can affect previous work based on line search optimization.
We can use this robustness property to unflip parametrizations produced by other
methods that do not guarantee flip-free minimizers, such as \(\EARAP\).
In \figref{unflippingtriangles} we initialize our optimization with the
result of an \(\EARAP\) UV map, and run it until the map contains no
flipped triangles, unflipping the triangles left over by \(\EARAP\). 
The limits of this robustness are discussed in \secref{limitations}.

\subsubsection*{Closed-Form Solutions for Every Substep}

The four substeps of our method are explicitly computable with
closed-form solutions for the energies \(\ESG\) and \(\ESD\)
(unlike, e.g., \cite{Brown2021}) and do not require parameter tuning.
This makes our method easy to implement.

\subsubsection*{Efficient Evaluation \& Parallelization}

Every step of our method can be computed efficiently.
The \(\bW\) step contains only a single linear solve, and the matrix \(L\) does
not change as long as the penalties \(\mu_i\) do not change.
This allows us to precompute the decomposition once, and only apply a cheap
backsubstitution every iteration.
To achieve this, we use Suitesparse's CHOLMOD if the constraints make the
problem in \(\bW\) definite, and Suitesparse's UMFPACK if the constraints result
in an indefinite problem \cite{suitesparse}.
The \(\bU\), \(\bP\) and \(\bLambda\) steps all decouple over elements, and can
thus be computed for each triangle/tetrahedron separately, in parallel.
This makes the algorithm highly parallelizable.
We implement parallelization using OpenMP.

\section{Convergence Analysis}
\label{sec:convergence}
Even though the proposed Algorithm \ref{alg:mainalgorithm} looks like the usual ADMM scheme, there is an important caveat:  the constraints in our formulation \eqref{eq:admm-os} is \emph{non-linear} and \emph{non-convex}.
As a result, standard ADMM convergence analysis \cite{Boyd2011}
does not apply. This distinguishes our method from many other computer graphics ADMM methods
discussed in \secref{relatedaugmentedlagrangian}.
As \equref{admm-os} is non-convex, a natural question is whether the ADMM will converge or not. 
We confirm that this is true in this section. 
Specifically, we show that, under certain conditions (which in practice often hold), the sequence \(\{(\bW^{(k)},\bU^{(k)},\bP^{(k)}, \bLambda^{(k)})\}_{k \ge 0 }\) generated by proposed ADMM method will
converge to a Karush–Kuhn–Tucker (KKT) point
\((\bW^{*},\bU^{*},\bP^{*},\bLambda^{*})\) of \equref{admm-os} that is defined by
the conditions
\begin{equation}\label{eq:kkt}
\left \{
\begin{aligned}
& w_i \nabla f(\bP_i^*) -\mu_i{\bU_i^*}^\transp\bLambda_i^*=0, \, \forall i,\\
& \partial g({\bU_i}^*)- \mu_i\bLambda_i^*{\bP_i^*}^\transp\ni 0, \, \forall i ,\\
& (G\bW^*)_i = {\bU_i}^* {\bP_i}^*, \,  \forall i.
\end{aligned}
\right.
\end{equation}
Here, $g(\bU_i)$ is the indicator function over the rotation matrix set $\SOn{d}$ and $\partial g(\cdot)$ is the limiting subdifferential. We refer to \cite{li2020understanding} for a recent review of stationarity in non-convex problems for the concrete definition.
The second condition of the KKT system can be written as
\({\bU_i}^* = \argmin_{\bU_i \in \SOn{d}} \left(\bU_i \cdot \bLambda_i^*{\bP_i^*}^\transp \right)\). 

We now present the convergence condition.
First, we introduce an explicit boundedness condition on the gradient
norm of \(f(\cdot)\) with respect to the sequence \(\{\bP^{(k)}\}\).

\begin{condition}\label{cond:gradf}
	The gradient of \(f\) evaluated at the iterates \((\bP^{(k)})_{k\ge0}\)
    generated by our ADMM algorithm is bounded, i.e.,
    \(\|\nabla f(\bP_i^{(k)})\| \leq B_i,~\forall i\).
 \end{condition}
This condition must be verified by the user when employing the method.
Such boundedness conditions are common for nonconvex ADMM
\cite{Zhang2021,zhang2019fundamental,gao2020admm}.
\figref{uvmappingbounds} shows that this bound holds in practice for a
variety of UV parametrization experiments, but it does not have to
hold for every problem (see \secref{limitations}).
This is because the energy function \(f\) is not globally, but only locally
Lipschitz continuous. 
However, the usual convergence results for general non-convex ADMM
(e.g., even with linear constraints) rely on a global Lipschitz
condition, which is the key for obtaining a bounded sequence for both primal and
dual variables.
Without global Lipschitz continuity it is challenging to bound the difference
between two iterates.
To address this, we build our convergence analysis on \condref{gradf}, which not
only enables us to use a local Lipschitz property, but also plays an important
role in providing us with an explicit bound for the penalty parameter. 
The bound is crucial for our practical implementation as well, as
elaborated in \secref{algorithm}. 

Our second condition lets us bound the difference between
\(\Lambda_i^{(k+1)}\) and \(\Lambda_i^{(k)}\).
\begin{condition}\label{cond:lambdabound}
    There exists a \(\gamma > 0\) such that 
    \begin{equation*}\begin{split}
        \norm{\Lambda_i^{(k+1)} - \Lambda_i^{(k)}}^2
        \leq \frac{\gamma}{4}
        \Big\lVert &\Lambda_i^{(k+1)}
        + U_i^{(k+1)}(\Lambda_i^{(k+1)})^\transp U_i^{(k+1)} - \Lambda_i^{(k)}
        - U_i^{(k)}(\Lambda_i^{(k)})^\transp U_i^{(k)} \Big\rVert^2
        \;\textrm{.}
    \end{split}\end{equation*}
\end{condition}
This condition is needed to be able to bound the \(\Lambda\) update by
the \(U\) and \(P\) updates.
Because of the symmetrization in the \(P\) update \equref{rootfindinginp}, we
lose information on the antisymmetric part of \(\Lambda\) -- this additional
condition allows us to control the antisymmetric part using only the symmetric
part.
\figref{uvmappingbounds} shows appropriate bounds \(\gamma\) for a few UV
parametrization applications.
We conjecture that \condref{lambdabound} may not be necessary to prove
convergence.

We leave the investigation of convergence behavior with weaker conditions to
future work.
These conditions are needed to enable our particular proof;
different proof strategies could use other conditions.
For many practical applications, however, our conditions are reasonable to
impose.
\figref{uvmappingbounds} shows that these conditions are realistic for our
UV parametrization examples presented in \figref{uvmapping}.
Both the gradient bound from \condref{gradf}, as well as the
\(\Lambda\) ratio from \condref{lambdabound} can be bounded during the
optimization.
\\

\begin{figure}
    \includegraphics[width=\textwidth]{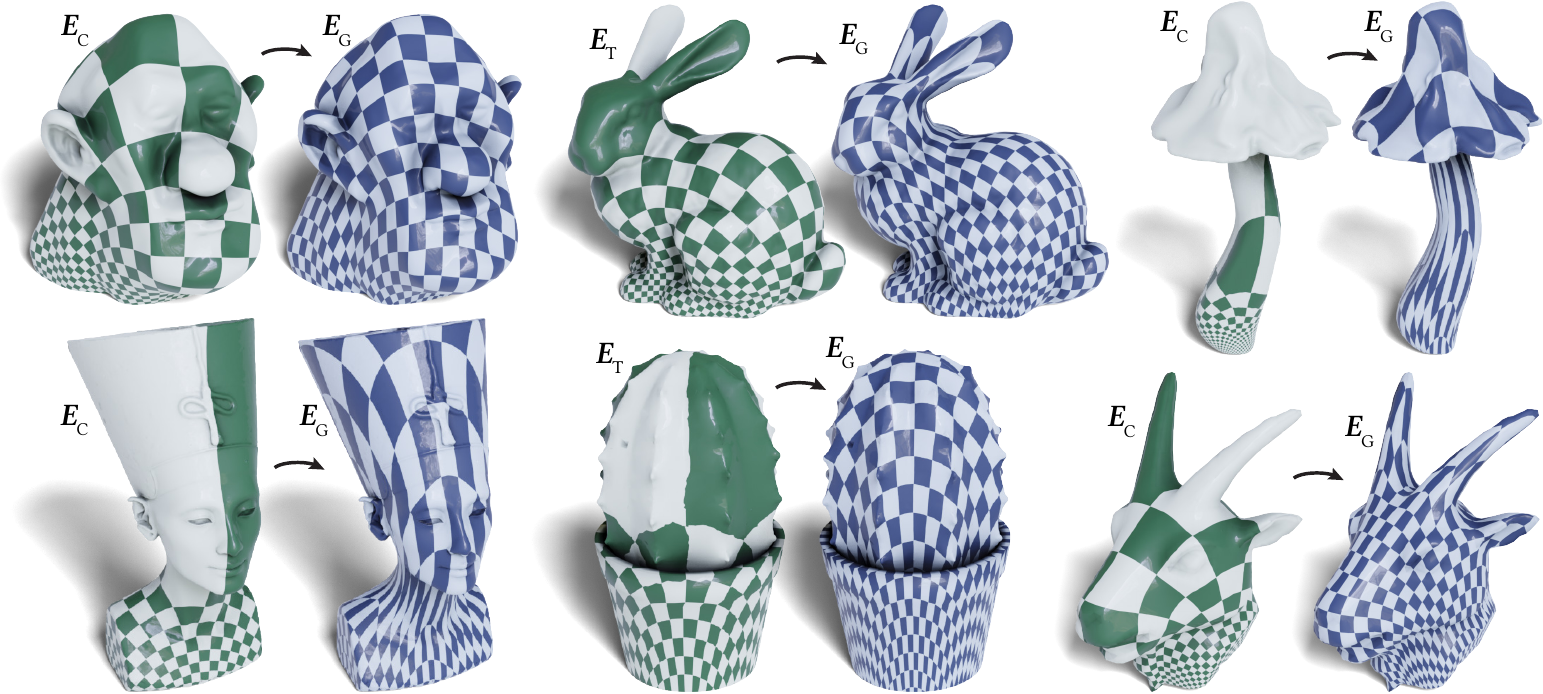}
    \caption{Optimizing \(\ESG\) with our splitting scheme to compute
    UV maps for a variety of surfaces.
    The surfaces are textured with a regular checkerboard texture.
    All generated UV maps are flip-free.
    The errors over the runtime of the optimizations are displayed in
    \figref{uvmappingbounds}.
    \label{fig:uvmapping}}
\end{figure}

Since ADMM algorithms are primal-dual methods, the crux of our convergence analysis is to use the primal variable $(\bW,\bU,\bP)$  to bound the dual update of $\bLambda$, leading to a sufficient decrease in the augmented Lagrangian function. To start, we derive an explicit local Lipschitz constant for various deformation energies, which is central to our convergence analysis. 
\begin{lemma}\label{lm:lipc}
    We have
	\[ \| \nabla f(\bP_i^{(k+1)}) -\nabla f(\bP_i^{(k)})\| \leq F_i \|\bP_i^{(k+1)}-\bP_i^{(k)}\|, \, \forall i.\]
\begin{itemize}
	\item For $f = \fSG$, we have $F_i = \left(1+\frac{\sqrt{d}}{{ C_i^{\text{L}} }^2}\right)$. 
	\item For $f = \fSD$, we have $F_i = \left(1+ \frac{3\sqrt{d}}{{ C_i^{LG}}^4}\right)$. 
\end{itemize}
The constants are given by
\begin{itemize}
    \item $C_i^L =-\frac{B_i}{2} +\frac{\sqrt{4+B_i^2}}{2} $; and
    \item
$C_i^{LG}$ is the positive root of the quartic equation, i.e., $x^4+B_ix^3-1=0$.
\end{itemize}
Moreover, $C_i^{L} \leq C_i^{LG} \leq 1$. 
\end{lemma}
Based on Lemma \ref{lm:lipc} establishing an explicit local Lipschitz constant,
we can now derive a few basic properties of our ADMM algorithm. 
\begin{proposition}[Sufficient decrease property]
	\label{prop:admm}\\
Suppose that \(\mu_i>  \frac{1}{2}\left(-(w_i-2\epsilon) + \sqrt{(w_i-2\epsilon)^2+16\gamma w_i^2F_i^2}\right)\), \(h_i \ge \frac{4\gamma w_i^2B_i^2}{{\mu_i}} + 2\epsilon\), and \(0<\epsilon<\frac{\min_iw_i}{2} \).
Let \(\{(\bW^{(k)}, \bU^{(k)}, \bP^{(k)}, \bLambda^{(k)})\}_{k= 0}^\infty\) be
the sequence of iterates generated by our ADMM algorithm, and denote
\(\Phi(\bW^{(k)}, \bU^{(k)}, \bP^{(k)}, \bLambda^{(k)})\) by \(\Phi^k\).
 Then
\begin{equation*}
 	\begin{aligned}
 	\Phi^{k+1}-\Phi^k \leq& -\frac{1}{2}\lambda_{\text{min}}(L)\|
 	\bW^{(k+1)}-\bW^{(k)}\|^2\\
 	&-\sum_{t=1}^m \epsilon\left(\|\bU_i^{(k+1)}-\bU_i^{(k)}\|^2+\|\bP_i^{(k+1)}-\bP_i^{(k)}\|^2\right). 
 	\end{aligned}
 	\end{equation*}
\end{proposition}

\begin{figure}
    \includegraphics[width=\linewidth]{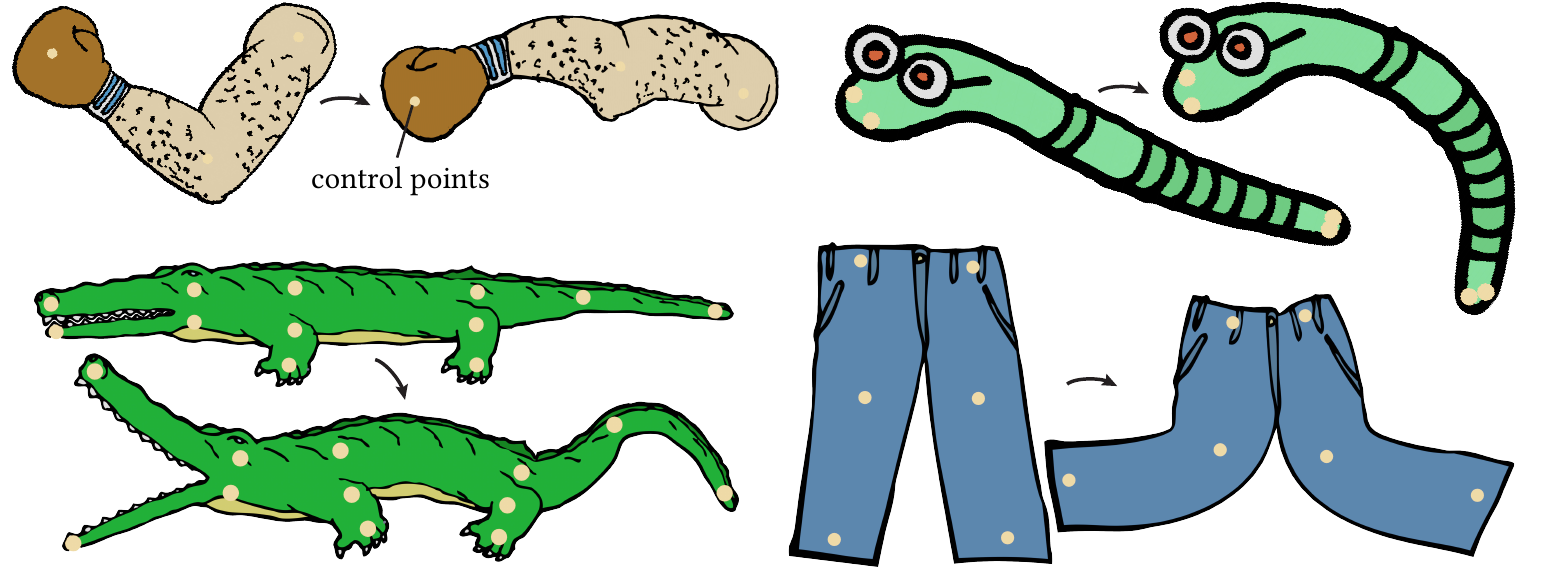}
    \caption{Computing low-distortion deformation of surfaces in \(\Ro^2\)
    by minimizing \(\ESD\), initializing with the identity map.
    The highlighted control points in the target mesh are fixed to the desired
    position, and our method is employed to minimize distortion.
\label{fig:deformationtwod}}
\end{figure}

We are now ready to prove a global convergence result for our ADMM algorithm
by characterizing the cluster point of the generated sequence.
\begin{theorem}[Global convergence of our splitting method]
	\label{main-theorem}
    If the set of KKT solutions for \equref{admm-os} that satisfy \equref{kkt}
    is non-empty, then
    the augmented Lagrangian function \(\Phi(\bW,\bU,\bP,\bLambda)\) is a
    Kurdyka-\L{}ojasiewicz function, and hence the sequence generated by
    \algoref{mainalgorithm}, \linebreak
    \(\{(\bW^{(k)}, \bU^{(k)}, \bP^{(k)}, \bLambda^{(k)})\}_{k=0}^\infty\),
    converges to a KKT point of \equref{admm-os}.
\end{theorem}
For a discussion on Kurdyka-\L{}ojasiewicz functions,  we refer the reader to
\cite{Bolte2007} for details.
The proofs for the statements in this section are given in
\appref{appendixproof}.

\section{Implementation}
\label{sec:algorithm}

\algoref{mainalgorithm} is the basic algorithm underlying our method, and we can
analyze its properties theoretically (see \secref{convergence}).
The method we implement in code is slightly different, taking practical
considerations into account.

\subsubsection*{Termination Condition}
There is no termination condition provided in \algoref{mainalgorithm}.
We support a variety of termination conditions, such as a target energy, or
no flipped triangles present (see \figref{unflippingtriangles}).
The general-purpose termination condition, which we use in all experiments
unless indicated, is a modified version of the primal and dual augmented
Lagrangian errors \cite[Section 3.3.1]{Boyd2011}, adjusted for our setting:
\begin{equation}\begin{split}\label{eq:primaldualerror}
    \left(e^{\textrm{prim}}_i\right)^2 &=
    \norm{(G\bW^{(k)})_i - \bU^{(k)}_i \bP^{(k)}_i}^2 , \\
    \left(e^{\textrm{prim}}\right)^2 &=
    \sum_{i=1}^m \left(e^{\textrm{prim}}_i\right)^2 , \\
    \left(e^{\textrm{dual}}_i\right)^2 &=
    \mu_i^2 \norm{(G\bW^{(k)})_i - (G\bW^{(k-1)})_i}^2 \\
    \left(e^{\textrm{dual}}\right)^2 &=
    \sum_{i=1}^m \left(e^{\textrm{dual}}_i\right)^2
    \;\textrm{.}
\end{split}\end{equation}
We terminate the method if both of these errors are below the thresholds
\begin{equation}\begin{split}\label{eq:primaldualthresholds}
    e^{\textrm{prim}} &< \varepsilon_\textrm{abs} \sqrt{dm} +
    \varepsilon_\textrm{rel}
    \max\left\{\norm{G\bW^{(k)}}, \norm{\bP^{(k)}}\right\} , \\
    e^{\textrm{dual}} &< \varepsilon_\textrm{abs} \sqrt{dm} +
    \varepsilon_\textrm{rel} \norm{G^\transp \bLambda^{(k)}}
    \textrm{.}
\end{split}\end{equation}
For most examples throughout this article, we choose
\(\varepsilon_\textrm{abs} = 10^{-6}, \varepsilon_\textrm{rel} = 10^{-5}\).
For deformation experiments, we use
\(\varepsilon_\textrm{abs} = 5 \cdot 10^{-10}, \varepsilon_\textrm{rel} = 5 \cdot 10^{-9}\).
Additionally, we add to the termination condition that all elements in the
iterate \(\bW^{(k)}\) have to be flip-free.

\begin{figure}
    \includegraphics[width=\linewidth]{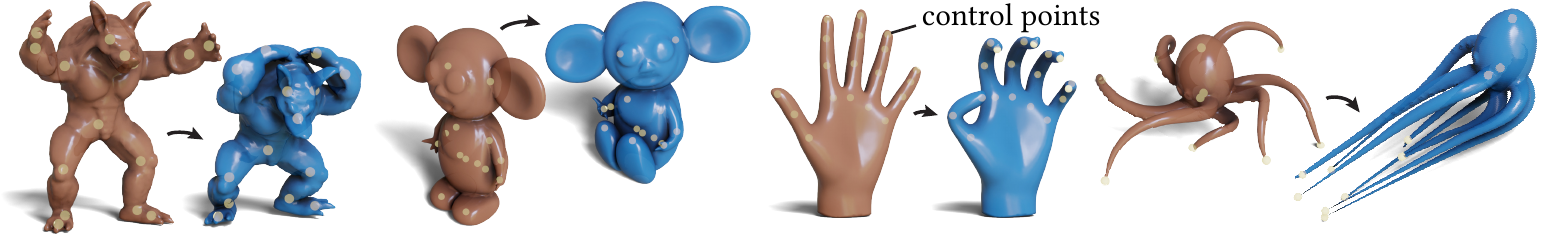}
    \caption{Computing low-distortion deformation of volumes in \(\Ro^3\)
    by minimizing \(\ESD\), initializing with the identity map.
    The highlighted control points in the target mesh are fixed to the desired
    position, and our method is employed to minimize distortion.
\label{fig:deformationthreed}}
\end{figure}

\subsubsection*{Rescaling \(\mu_i\) and \(\bLambda\)}
To speed up the optimization, we dynamically adjust the penalty
parameters \(\mu_i\)
\cite[Section 3.4.1]{Boyd2011}.
The goal of the rescaling algorithm is to keep \(e^{\textrm{prim}}\) and
\(e^{\textrm{dual}}\) from \equref{primaldualerror} roughly equal.
Thus,
\begin{itemize}
    \item if \(e_i^{\textrm{prim}} > \rho e_i^{\textrm{dual}}\), we multiply
    \(\mu_i\) by \(\frac{\rho}{2}\) and divide \(\bLambda^{(k)}_i\) by
    \(\frac{\rho}{2}\);
    \item if \(e_i^{\textrm{dual}} > \rho e_i^{\textrm{prim}}\), we divide
    \(\mu_i\) by \(\frac{\rho}{2}\) and multiply \(\bLambda^{(k)}_i\) by
    \(\frac{\rho}{2}\),
\end{itemize}
where, in our implementation, we set \(\rho = 5\).
We employ a lower bound for \(\mu_i\) of \(\frac12 \mu_{\min}\),
where \(\mu_{\min}\) is the bound from \secref{convergence}
computed with \(B_i = \sqrt{5(1 + \normu{\nabla f(\bP^{(k)}_i)}^2)}\),
and \(F_i^2\) is capped at \(\varepsilon_m^{-1/4}\)
(\(\varepsilon_m\) is the floating point machine epsilon).
The changing
\(\mu_i\) and \(B_i\) as well as
the lower bound of \(\frac12 \mu_{\min}\)
do not reflect the conditions of our convergence proof,
but such heuristic modifications to ADMM algorithms for actual
implementations are used in practice:
See, e.g., \cite[Section 3.4.1]{Boyd2011} for rescaling,
and \cite[Section 6.1]{Vaswani2019} for an example of gradient descent step
sizes larger than suggested by the Lipschitz continuity bound
(which is also where our \(\mu_{\min}\) originates).

Since changing the penalties \(\mu_i\) requires decomposition of
\(L\), we rescale sparingly:
five times directly after initialization, and
every \(5 (\frac32)^p\) iterations, where \(p\) is the number of past
    rescaling events.
This way we rescale more in the beginning of the optimization, when the iterates
are changing a lot, and less afterwards, when the iterates are not changing as
much anymore.
The initial penalty parameters (before rescaling) are set to
\(\mu_i = w_i\).
\(h_i\) is set to its value from \secref{convergence}, and disabled if we expect
many flipped triangles in the input.
\(\gamma\) is always set to \(1\) and \(\epsilon\) to \(0\).

\subsubsection*{Initializing \(\bU, \bP, \bLambda\)}
The user does not need to supply \(\bU^{(0)}, \bP^{(0)}, \bLambda^{(0)}\).
We can use the supplied \(\bW^{(0)}\) to initialize \(\bU, \bP\) by employing
a polar decomposition,
\begin{equation}\label{eq:initialpolardecomposition}
    (G \bW^{(0)})_i = \bU^{(0)}_i \bP^{(0)}_i
    \;\textrm{.}
 \end{equation}
This is, in practice, computed using the singular value decomposition (svd).
For \linebreak \((G \bW^{(0)})_i = R_1 \Sigma R_2^\transp\), we set
\begin{equation}\begin{split}\label{eq:polardecompositionapplied}
    \bU^{(0)}_i &= R_1 R_2^\transp \\
    \bP^{(0)}_i &= R_2 \Sigma R_2^\transp
    \;\textrm{.}
\end{split}\end{equation}
If \((G \bW^{(0)})_i\) is flipped, then \(R_1 R_2^\transp\) is not a
rotation matrix. In this case we multiply its last column by \(-1\) such
that \(\det\bU^{(0)}_i = 1\), and set \(\Sigma = \varepsilon I\).
To avoid numerical problems, we also set all values on the diagonal of
\(\Sigma\) to \(\varepsilon\) should they be smaller.
We set \(\varepsilon = \varepsilon_m^{1/4}\) when optimizing \(\ESG\), and
\(\varepsilon = \varepsilon_m^{1/8}\) when optimizing \(\ESD\), where
\(\varepsilon_m\) is the machine epsilon.
\(\Lambda^{(0)}\) is set to \(0\).
\\

\begin{figure}
    \includegraphics[width=\linewidth]{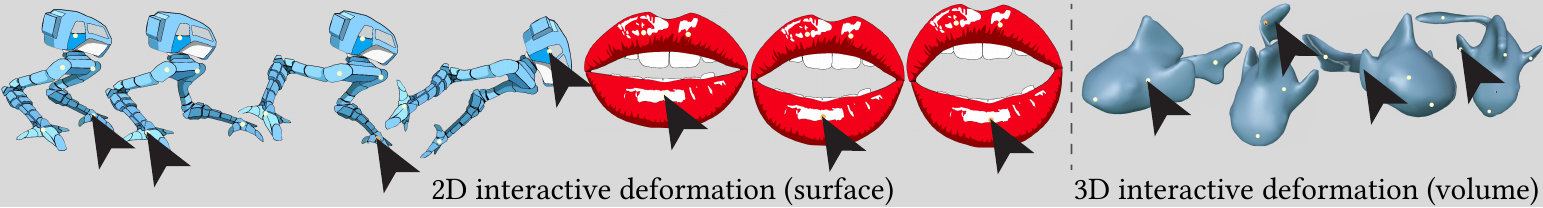}
    \caption{
    Since each step of our iteration method can be
    evaluated cheaply, the method is well-suited for interactive deformation.
    In our tool, the user picks which vertices they would like to constrain, and
    then drags them around while the shape deforms.
    For a video of this interactive application, see supplemental material.
    \label{fig:interactivedeformation}}
\end{figure}

\begin{algorithm}
\caption{Our implementation of \algoref{mainalgorithm}
\label{alg:actualalgorithm}}
\begin{adjustwidth}{-0.5em}{}
\begin{algorithmic}[1]
    \Method{SplittingOptimization}{\bW^{(0)}}
    \State \(\bU^{(0)}, \bP^{(0)} \gets
    \textrm{polar\textunderscore{}decomposition} ( \bW^{(0)} ) \)
    \State \(\bLambda^{(0)} \gets 0\)
    \State \(\mu_i \gets \max(\mu_{\min,i}, w_i) \quad \forall i\)
    \State \(\operatorname{decompose}(L)\)
    \For{\(k \gets 1, \dots, \textrm{max\textunderscore{}iter}\)}
        \State \(\!\! \bW^{(k)} \gets \argmin_{\bW\!, A\bW=b} \,
        \Phi\left(\bW, \bU^{(k-1)}, \bP^{(k-1)}, \bLambda^{(k-1)}\right) \)
        \State \(\bU^{(k)} \gets \argmin_{\bU} \,
        \Phi\left(\bW^{(k)}, \bU, \bP^{(k-1)}, \bLambda^{(k-1)}\right)
        \, + \, p\left(\bU,\bU^{(k-1)}\right) \)
        \State \(\, \bP^{(k)} \gets \argmin_{\bP} \,
        \Phi\left(\bW^{(k)}, \bU^{(k)}, \bP, \bLambda^{(k-1)}\right) \)
        \State \(\mkern 1mu \bLambda^{(k)}_i \gets \bLambda^{(k-1)}_i
        +\, ( G \bW^{(k)} )_i - \bU_i^{(k)} \bP_i^{(k)} \quad \forall i\)
        \If{rescale\textunderscore{}at\textunderscore{}iter(k)}
            \State \(\mu_i \gets \operatorname{rescale}(\mu_i,
            \bW^{(k)}, \bU^{(k)}, \bP^{(k)}, \bLambda^{(k)})
            \quad \forall i\)
            \State \(\operatorname{decompose}(L)\)
        \EndIf
        \If{termination\textunderscore{}condition\((\bW^{(k)},
        \bU^{(k)}, \bP^{(k)}, \bLambda^{(k)})\)}
            \Break \medskip
        \EndIf
    \EndFor
\end{algorithmic}
\end{adjustwidth}
\end{algorithm}

Our complete practical method is described in pseudocode form in
\algoref{actualalgorithm}.
We use IEEE double precision as our floating point type.
The actual \texttt{C++} implementation built on libigl \cite{libigl}
will be publicly released under an open-source license after publication.
We run our implementation on a 2.4GHz Quad-Core Intel i5 MacBook Pro with 16GB
RAM.

The user needs to initialize our method with a map to a target mesh \(\bW^{0}\).
This map does not need to be flip-free,
however it can not be arbitrary (see \secref{limitations}).
We initialize, depending on the application, either with minimizers of
\(\ETutte\) or \(\EConformal\), which can be optimized very cheaply.
\figref{uvmappingbounds} shows the primal and dual errors over the entire
optimization, as well as proxies for Conditions \ref{cond:gradf} \&
\ref{cond:lambdabound}, for applications of
\algoref{actualalgorithm} to UV parametrization.

\section{The Symmetric Gradient Energy}
\label{sec:symmetricgradient}

As a brief aside in the larger story of our optimization algorithm, we propose an
alternative to the distortion energies mentioned in
\appref{deformationenergies}, which we call the
\emph{symmetric gradient energy}.
Although our main method applies to a broad class of distortion energies,
we find that this alternative energy yields maps with favorable properties.

\begin{figure}
    \includegraphics[width=\linewidth]{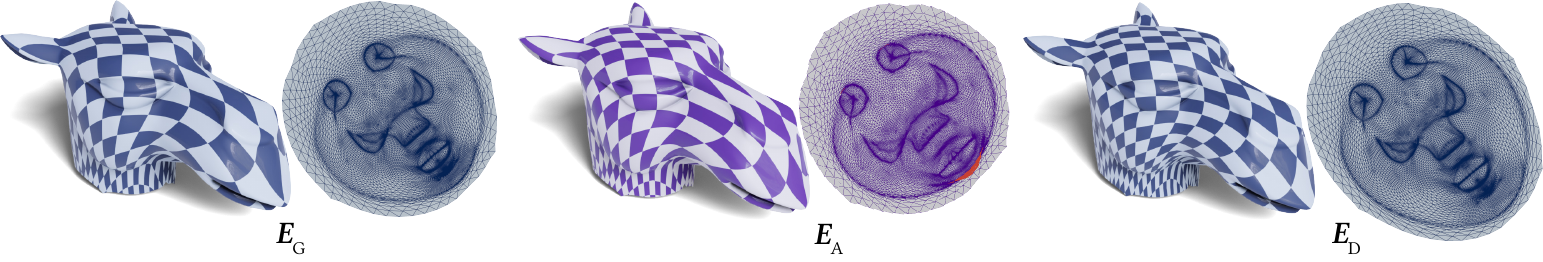}
    \caption{\(\ESG\) \figloc{left} exhibits more qualitative similarity to
    \(\EARAP\) \figloc{center} than \(\ESD\) \figloc{right}.
    The \(L^2\) area distortions are, from left to right,
    \(0.000120\), \(0.000103\), \(0.000213\)
    (flipped triangles in red).
    \label{fig:symmetricgradientandarap}}
\end{figure}

While \(\fSD\) is
symmetric with respect to inversion of its input matrix \(X\), its gradient
\equref{symmetricdirichletdefininggrad} is not.
Moreover, the singularity in the gradient is rather strong (\(\sim 1/x^3\)).

\begin{definition}[Symmetric gradient energy]
	The defining function \(\fSG : \Rnm{d}{d} \rightarrow \Ro\) of the symmetric
    gradient energy is given by
	\begin{equation}\label{eq:symmetricgradientdefiningfunction}
		\fSG(X) \coloneqq \frac12 \norm{X}^2 - \log\det X
		\;\textrm{.}
	\end{equation}
    The symmetric gradient energy \(\ESG\) is defined as the flip-free generic
    distortion energy \equref{genericdeformationenergy} with \(f=\fSG\).
\end{definition}

Just like \(\ESD\), \(\ESG\) is continuous and bounded for all flip-free
maps.
\(\ESG\) is singular exactly when \(\ESD\) is, just with a weaker singularity of
\(\sim \log x\) instead of \(\sim 1/x\) (see also \appref{symmetricdirichlet}).
\(\ESG\) is also invariant to rotations, as for any rotation matrix \(U\),
\(\fSG(UX) = \fSG(X)\).

To our knowledge, \(\ESG\) has not appeared in previous work on distortion
energies in this form.
\(\ESG\) is, however, similar to other alternatives, chiefly
among them the norm of the Hencky strain tensor \cite{Hencky1928}
\(\norm{\log X^T X}^2\), which also features a logarithmic term.
Other previous works discuss a strain energy that consists of only the logarithmic
term of \(\fSG\) \cite{Neff2016}.
\(\ESG\) is, of course, also similar to \(\ESD\), which replaces the logarithmic
term of \(\ESG\) with an inverse term.

The gradient of \(\ESG\)'s defining function is given by
\begin{equation}\label{eq:symmetricgradientdefininggrad}
	\nabla\fSG(X) = X - X^{-\transp}
	\;\textrm{.}	
\end{equation}
As \(X\) appears both as itself and as its inverse in the gradient of
\(\fSG\), we call this energy the symmetric gradient energy to parallel the
symmetric Dirichlet energy, where both appear in the function itself.
The singularity in the gradient (\(\sim 1/x\)) is weaker than \(\ESD\)'s
(\(\sim 1/x^3\)).

\(\ESG\) can be an alternative to other distortion energies in certain
applications, depending on specific goals.
Minimizers of \(\ESG\) look more visually similar to the popular non-flip-free
minimizers of \(\EARAP\) than minimizers of \(\ESD\).
If the goal is a maximally rigid map that is still flip-free, \(\ESG\) is a
valid choice.
Additionally, minimizers of \(\ESG\) can exhibit lower
area distortion than minimizers of \(\ESD\), yielding a more faithful map in
terms of area (see \figref{symmetricgradientandarap}).
Because of these properties, we
prefer using \(\ESG\) for UV parametrization and volume correspondence
(where the lower area distortion is a key feature),
and \(\ESD\) for deformation (where the weaker singularity of \(\ESG\) can
lead to distorted elements near the constrained parts of the mesh).

\section{Results}
\label{sec:results}

We use our splitting scheme to optimize distortion energies
for three different applications:
UV mapping, shape deformation, and volume correspondence.

\subsection{UV Maps}
\label{sec:uvmaps}

A UV map of a triangle mesh \(\bV \subseteq \Ro^3\) is a map from
\(\bV\) into \(\Ro^2\), the UV space.
UV maps have a variety of applications, such as
texturing surfaces \cite[Section 6.4]{Akeley2013},
quad meshing \cite{Bommes2009},
machine learning on meshes \cite{Maron2017},
and more \cite{Pal2014,Schmidt2019}.
We can compute a UV map by minimizing the distortion of a map from \(\bV\)
into UV space.
For applications such as texture mapping, it is
especially
important to have a low-distortion UV map, since high distortion
will require higher-resolution images for texturing.
Figures \ref{fig:teaser}, \ref{fig:uvmapping} show our splitting method used to
minimize \(\ESG\) to arrive at a UV map.
The surfaces are textured using a regular checkerboard texture to
visualize the distortion of the UV map;
the rendered checkerboard scale is manually set to attain qualitatively similar
triangle sizes.
In \figref{symmetricdirichletinoursandprevious} and \tableref{comparisonwithprevious}
we compare our method with multiple previous works when optimizing \(\ESD\).
We can see that for some examples the meshes are so challenging (as they
contain a lot of branches and appendages that get squished down into a small
area) that the previous methods were unable to produce a flip-free
distortion minimizing optimization -- an area where our method's robustness to
flipped triangles enables us to compute results despite the difficulty.
\figref{dataset} shows our method applied to a large parametrization data set.

\subsection{Deformation}
\label{sec:deformation}

\begin{figure}
    \includegraphics[width=\textwidth]{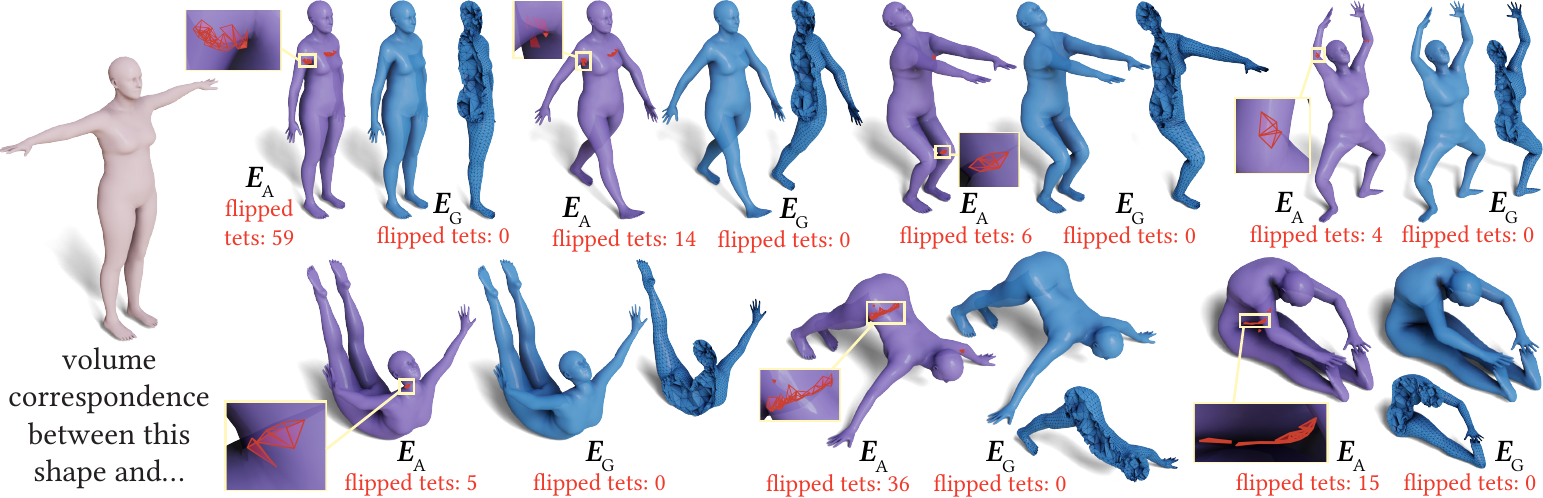}
    \caption{Computing a correspondence between the tetrahedral mesh of
    a human \figloc{far left}, and the interior of a variety of surfaces
    whose boundaries correspond to the original human.
    Correspondences computed with our method using \(\ESG\) exhibit no flipped
    tetrahedra \figloc{top},
    while correspondences computed using \(\EARAP\) can contain
    flipped tetrahedra \figloc{bottom, highlighted in red}.
    \label{fig:volumecorrespondence}}
\end{figure}

Distortion energies can be used for deformation by constraining
a part of the target mesh \(\bW\).
We can constrain isolated vertices, as well as entire
regions of the mesh, and then run our optimization method,
initialized with the identity map.
The few initially distorted (and potentially flipped)
elements do not prevent our method from finding a solution.

In Figures \ref{fig:teaser}, \ref{fig:deformationtwod}, \ref{fig:deformationthreed}
we deform a variety of surfaces and volumes by constraining vertices of
the target mesh, and compare the results of the deformation using our method using \(\ESD\)
with the results of minimizing the ARAP energy \(\EARAP\).
Our method succeeds in producing natural-looking deformations
while avoiding inverted elements.
\figref{interactivedeformation} shows our method applied to interactive deformation:
we created a user interface that allows fixing vertices of the
target mesh and dragging them to the desired location.
The displayed mesh is updated interactively during each iteration of our method.
This results in a smooth interactive experience,
as each iteration of our method is cheap to compute.

\subsection{Volume Correspondence}
\label{sec:volumecorrespondence}

Our method can be used to compute a correspondence between
the interior volumes of two given surfaces.
To do this, we minimize the distortion of a map between the interior
volume of the first surface (which has been tet-meshed
\cite{Hu2018,Hang2015}), and the same volume with
its boundary fixed to the second surface.
The optimization is initialized with a minimizer of
\(\ETutte\), which can contain flipped tetrahedra,
however our method is able to arrive at a flip-free distortion-minimizing
volume correspondence map nevertheless.

\figref{volumecorrespondence} shows our method applied to
compute correspondences between the interior of a human and a variety of
other surfaces that show the same human, but in a different position.
\figref{teaser} shows a volume correspondence between two different
configurations of a teddy bear.
In both cases, our method produces flip-free distortion-minimizing
volume correspondences.
Surface correspondences for applying our method
can be computed, e.g., using \cite{Ezuz2019}.

\begin{figure}
    \includegraphics[width=\textwidth]{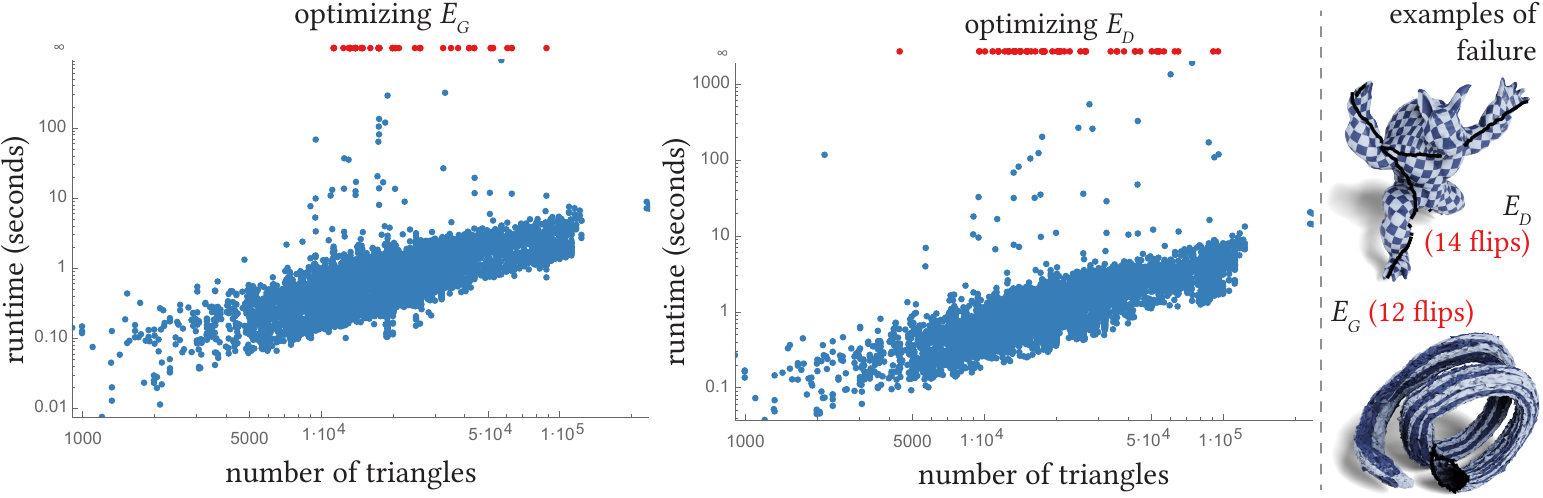}
    \caption{Computing UV parametrizations for all meshes in the D1, D2 and D3
    datasets \cite{Liu2018} (modified to exclude degenerate meshes), consisting
    of 15525 meshes.
    The method was initialized with \(\ETutte\), run with a maximal iteration
    count of 100000, and with termination tolerances
    \(\varepsilon_\textrm{abs} = 5\cdot10^{-5}, \varepsilon_\textrm{rel} = 5\cdot10^{-4}\).
    Successful meshes are shown in blue, failed meshes
    (0.38\% for \(\ESG\), 0.53\% for \(\ESD\)) in red.
    \label{fig:dataset}}
\end{figure}

\section{Limitations}
\label{sec:limitations}

There are a few scenarios in which our method can fail.
We can not initialize our method with arbitrary input
(see \figref{limitations}, \emph{left}),
which can cause \(\|\nabla f(\bP_i^{(k)})\|\) to grow very large.
In this case, our method will not converge to a flip-free optimum;
this can prevent us from initializing with \(\EConformal\) in the presence of
too many flips (this is the case, e.g., for the tree example in
\tableref{comparisonwithprevious}).
In \figref{dataset}, a few inputs fail to yield a flip-free parametrization.

While minimizers of \(\ESD\) are guaranteed to be flip-free,
this does not mean that the maps will be bijective.
As with previous methods that employ \(\ESD\), one often obtains bijective maps
in practice, although this is not guaranteed.
A \(k\)-cover of a triangle mesh can be completely flip-free,
while also failing to be locally injective (see \figref{limitations},
\emph{center}).
This applies to all methods which optimize flip-free energies that are
not specifically bijective, such as \(\ESD\).
Related works propose solutions to this problem \cite{Smith2015,Garanzha2021}.

Our method can fail when constraints on \(\bW\) are imposed that
constitute a very large deformation from the initial target mesh \(\bW^{(0)}\).
An example of this can be seen in \figref{limitations}, \emph{right}.
The deformation application has the additional limitation that the termination
tolerances need to be set lower than for other applications to achieve good
results (although this can be somewhat remedied
by initializing with minimizers of \(\EARAP\).

Theorem \ref{main-theorem} guarantees convergence, assuming exact arithmetic.
The implementation on the computer uses floating-point arithmetic, which is not
exact, which can result in elements that are flipped for numerical reasons.
If such numerical issues occur while the optimization is still making
progress, the method might recover, like it does for maps with
flipped elements such as in \figref{unflippingtriangles}.
If this happens when the method can no longer make useful progress, the method
will fail to terminate
(this is a known issue with parametrization methods \cite{Shen2019}).

\section{Conclusion}
\label{sec:conclusion}

In this paper, we have proposed a new splitting scheme
for the optimization of flip-free distortion energies, discussed the convergence behavior of the resulting ADMM algorithm
under certain conditions, and demonstrated its utility in a variety of applications.

There are opportunities for future work in many directions.
On the application side, our method could be used for other
applications where flip-free distortion-minimizing mappings
are required, such as in elasticity simulation combined with contact mechanics,
for example in the context of an efficient solver such as
projective dynamics \cite{Bouaziz2014};
or to deform shapes with suitability for fabrication in mind \cite{Bouaziz2012}.
On the implementation side, our method could be considerably sped up by
implementing some of our parallelized instructions
in the \(\bP\) and \(\bU\) optimization step in a way that exploits
simultaneous execution capabilities of modern CPUs such as SSE or AVX,
similar related approaches \cite{McAdams2011} for the optimization of \(\EARAP\)
\cite[\texttt{polar\textunderscore{}svd3x3.h}]{libigl}.
On the algorithm side, further approaches to improve the performance
of the ADMM can be employed.
There is a lot of recent work on the topic, and some of it might
be able to speed up our splitting method.

\begin{figure}
    \includegraphics[width=\linewidth]{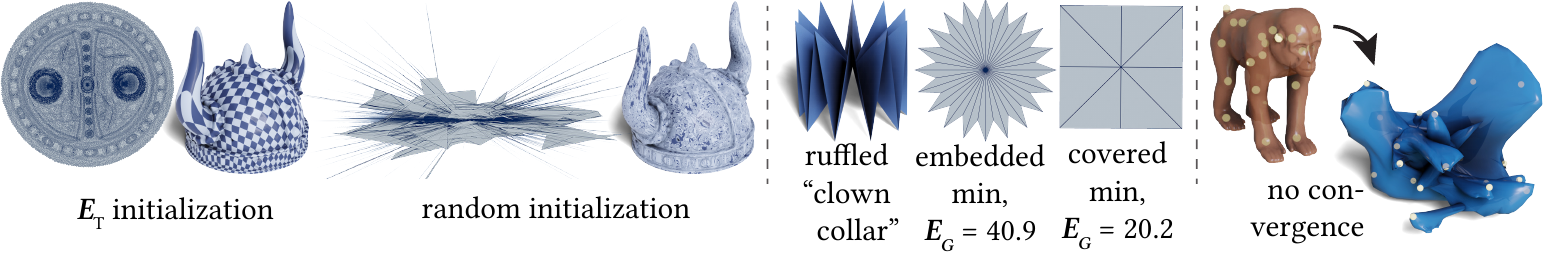}
    \caption{If our method is initialized with a bad initial target
    mesh \(\bW\), it will not converge like it would with a good initial
    target mesh (e.g., the minimizer of \(\ETutte\)) \figloc{left}.
    While our method will always produce a flip-free map,
    not all flip-free maps are bijective: the \(k\)-cover of this ruffled
    high-valence vertex has lower distortion than the bijective map \figloc{center}.
    If the target mesh \(\bW\) is constrained to a deformation that
    is very far from the initial mesh, our method can fail
    to converge \figloc{right}.
    \label{fig:limitations}}
\end{figure}

\begin{table}
\centering
    \includegraphics[width=\textwidth]{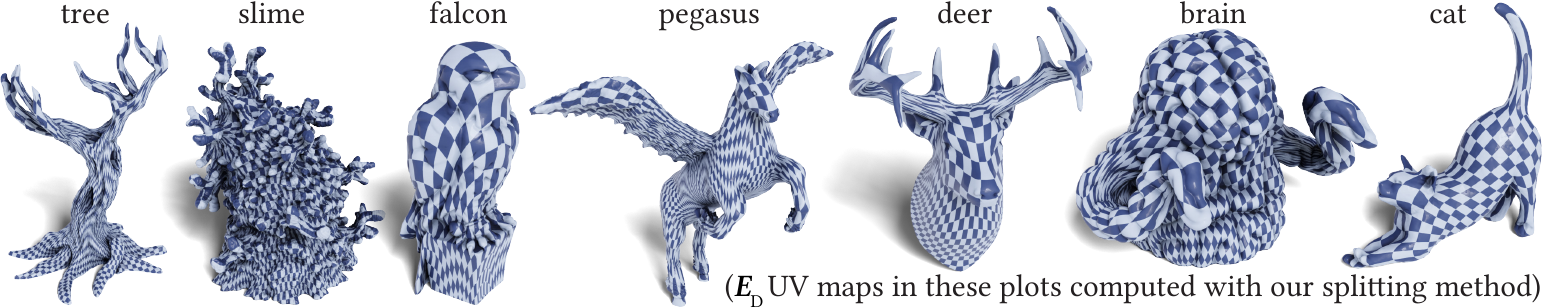}
    {\scriptsize\centering
\begin{tabular}{lrr|rr|rr|rr}
    \hline
    \vspace{-8pt}\quad& \multicolumn{1}{l}{\(\quad m\)} &
    \multicolumn{1}{l|}{\(\;\frac{\textrm{min.\ area}}{\textrm{max.\ area}}\)} &
    \multicolumn{1}{l}{SLIM} &
    \multicolumn{1}{l|}{ours} &
    \multicolumn{1}{l}{AKVF} & 
    \multicolumn{1}{l|}{ours} &
    \multicolumn{1}{l}{PP} & 
    \multicolumn{1}{l}{ours} \\
    & \multicolumn{1}{l}{} &
    \multicolumn{1}{l|}{} &
    \multicolumn{1}{l}{} &
    \multirow{2}{40pt}{\tiny (match \\ SLIM \(\ESD\))} &
    \multicolumn{1}{l}{} & 
    \multirow{2}{40pt}{\tiny (match \\ AKVF \(\ESD\))} &
    \multicolumn{1}{l}{} & 
    \multirow{2}{30pt}{\tiny (match \\ PP \(\ESD\))} \\
    & & & & & & & & \\
    \hline
        \rowcolor{clblue}
        camel & \(3.58\)k & \(0.00191\) &
        \(0.246s\) & \(0.134s^\dagger\) &
        \(0.947s\) & \(3.16s^\dagger\) &
        \(0.115s\) & \(3.20s^\dagger\) \\
        triceratops & \(5660\) & \(3.27 \cdot 10^{-7}\) &
        \(0.324s\) & \(0.233s^\dagger\) &
        \(0.968s\) & \(2.25s^\dagger\) &
        \(-^\ddag\) & \(0.403s^\dagger\) \\
        \rowcolor{clblue}
        cow & \(5.80\)k &\(0.00132\) &
        \(0.346s\) & \(0.212s^\dagger\) &
        \(1.53s\) & \(1.05s^\dagger\) &
        \(0.207s\) & \(1.08s^\dagger\) \\
        tooth & \(9129\) & \(3.27 \cdot 10^{-7}\) &
        \(0.426s\) & \(0.307s^\ast\) &
        \(1.36s\) & \(0.367s^\ast\) &
        \(0.252s\) & \(0.395^\ast\) \\
        \rowcolor{clblue}
        hand & \(9.36\)k & \(0.00174\) &
        \(-\) & \(1.95^\dagger\) &
        \(1.31s\) & \(1.51s^\dagger\) &
        \(0.516s\) & \(1.50s^\dagger\) \\
        deer & \(10.9\)k & \(0.000179\) &
        \(-\) & \(6.10^\ast\) &
        \(-\) & \(6.76^\ast\) &
        \(-^\ddag\) & \(240s^\dagger\) \\
        \rowcolor{clblue}
        horse & \(39.7\)k & \(0.000429\) &
        \(1.46s\) & \(1.23s^\dagger\) &
        \(7.64s\) & \(10.8s^\dagger\) &
        \(1.04s\) & \(10.9s^\dagger\) \\
        bread & \(49.9\)k & \(0.000760\) &
        \(2.07s\) & \(2.81s^\dagger\) &
        \(14.4s\) & \(2.84s^\dagger\) &
        \(1.10s\) & \(2.81s^\dagger\) \\
        \rowcolor{clblue}
        falcon & \(51.5\)k & \(0.000180\) &
        \(-\) & \(7.71s^\ast\) &
        \(-\) & \(7.48s^\ast\) &
        \(1.22s\) & \(9.05s^\ast\) \\
        cat & \(90.0\)k & \(0.00377\) &
        \(3.54s\) & \(19.4s^\ast\) &
        \(10.5s\) & \(15.7s^\ast\) &
        \(3.62s\) & \(25.4s^\ast\) \\
        \rowcolor{clblue}
        car & \(97.6\)k & \(5.25 \cdot 10^{-6}\) &
        \(3.83s\) & \(11.0s^\dagger\) &
        \(15.5s\) & \(11.2s^\dagger\) &
        \(1.99s\) & \(11.1s^\dagger\) \\
        brain & \(152\)k & \(0.00107\)
        & \(-\) & \(308s^\ast\) &
        \(-\) & \(308s^\ast\) &
        \(15.7s\) & \(16.8s^\ast\) \\
        \rowcolor{clblue}
        strawberry & \(313\)k & \(1.81 \cdot 10^{-6}\) &
        \(14.8s\) & \(55.3s^\ast\) &
        \(45.1s\) & \(59.8s^\ast\) &
        \(9.57s\) & \(60.5s^\ast\) \\
        slime & \(567\)k & \(0.000980\) &
        \(-\) & \(1140s^\ast\) &
        \(112s\) & \(38.4s^\ast\) &
        \(53.7s\) & \(42.9s^\ast\) \\
        \rowcolor{clblue}
        tree & \(630\)k & \(5.84 \cdot 10^{-5}\) &
        \(-\) & \(2080s^\ast\) &
        \(--\) & \(2090s^\ast\) &
        \(187s\) & \(76.0s^\ast\) \\
        pegasus & \(2390\)k & \(1.82 \cdot 10^{-10}\) &
        \(-\) & \(859s^\ast\) &
        \(-\) & \(866s^\ast\) &
        \(-\) & \(862s^\ast\) \\
    \hline
    \multicolumn{9}{l}{\vspace{7pt}
    \multirow{2}{400pt}{\(-\) a previous work was unable to find a distortion-minimizing flip-free mapping (it errored, or flips were present in the result)} }\\
    \multicolumn{9}{l}{\vspace{-5pt}
    \(--\) the algorithm terminated, but with a very high energy, which is counted
    as unsuccessful} \\
    \multicolumn{9}{l}{\vspace{-5pt}
    \(\ast\) our method initialized with a minimizer of \(\ETutte\)} \\
    \multicolumn{9}{l}{\vspace{-5pt}
    \(\dagger\) our method initialized with a minimizer of \(\EConformal\)} \\
    \multicolumn{9}{l}{\(\;\ddag\) previous method will sometimes work,
    and sometimes fail; thus reported as fail here} \\
    \hline
\end{tabular}
}
\quad\\ \quad\\     \caption{Comparing the runtime of UV maps generated
    with our method the the previous methods of SLIM \cite{Rabinovich2017},
    AKVF \cite{Claici2017}, and PP \cite{Liu2018}.
Since each of these related works uses their own termination condition, they
    do not all arrive at a parametrization with the same \(\ESD\).
    To compare against each previous method as intended by its authors, we run
    the publicly available implementation of the method until it satisfies its
    own termination condition (or is aborted after 150 minutes).
    We then measure \(\ESD\) of the previous method's UV map (after rescaling it
    to match the total area of the original mesh), and run our own algorithm to
    produce a flip-free map matching the previous method's \(\ESD\) up to a
    tolerance of \(10^{-6}\).
    Thus the runtimes in this table are plotted in pairs:
    a previous method, as well as our algorithm set to produce a map with the
    same distortion.
We initialize our method with minimizers of \(\ETutte\) and \(\EConformal\),
    and report the best time.
    If the previous method did not terminate successfully, our method is run
    with default termination conditions (which, in general, are set to produce
    lower-energy result that previous methods' termination conditions).
    Values are rounded to three significant digits.
    \label{tab:comparisonwithprevious}}
\end{table}

\section{Acknowledgements}
\label{sec:acknowledgements}

We thank Alp Yurtsever and Suvrit Sra for discussing ADMM proofs of convergence
with us.
Credit for meshes used goes to \cite{schlossbauer2020,Jacobson2013,MustangDave2015,Toawi2020,Catiav5ftw2015,blender,Holinaty2020,Stanford2020,Holinaty2020b,Nefertiti,paBlury2017,hugoelec2013,Baran2007,SMPL2021,Pauly2003,BlueRobot2017,GlossyRedLips2019,Crane2015,Sorkine2004,Liu2018,M3DM2020,schlossbauer2021,colinfizgig2013,AJade2019,hugoelec2013b,billyd2016,4MULE82016,YahooJAPAN2013}.

This work is supported by the Swiss National Science Foundation's Early
Postdoc.Mobility fellowship.
The MIT Geometric Data Processing group acknowledges the generous support of Army Research Office grants W911NF2010168 and W911NF2110293, of Air Force Office of Scientific Research award FA9550-19-1-031, of National Science Foundation grants IIS-1838071 and CHS-1955697, from the CSAIL Systems that Learn program, from the MIT--IBM Watson AI Laboratory, from the Toyota--CSAIL Joint Research Center, from a gift from Adobe Systems, from an MIT.nano Immersion Lab/NCSOFT Gaming Program seed grant, and from the Skoltech--MIT Next Generation Program.
 
\begin{appendices}
	\section*{Appendix}

\section{Distortion Energies}
\label{app:deformationenergies}

This appendix discusses distortion energies appearing in this article that have
been featured extensively in previous work.

\subsection{The Symmetric Dirichlet Energy}
\label{app:symmetricdirichlet}

\begin{definition}[Symmetric Dirichlet energy \cite{Smith2015}]
	The defining function \(\fSD : \Rnm{d}{d} \rightarrow \Ro\) of the symmetric
    Dirichlet energy is
	\begin{equation}\label{eq:symmetricdirichletdefiningfunction}
		\fSD(X) \coloneqq \frac12 \left(
		\norm{X}^2 + \norm{X^{-1}}^2 \right)
		\;\textrm{.}
	\end{equation}
    The symmetric Dirichlet energy \(\ESD\) is defined as the flip-free generic
    distortion energy \equref{genericdeformationenergy} with \(f=\fSD\).
\end{definition}

\(\ESD\) is finite and continuous on all maps
whose Jacobian has positive determinant.
By the definition of \equref{genericdeformationenergy}, a negative
determinant leads to infinite \(\ESD\) (due to \(\chi_+\)),
and since \(\fSD\) has a singularity for matrices with zero determinant,
there is a barrier preventing determinants from approaching zero.
Hence, minimizers of \(\ESD\) are flip-free (unless overconstrained).
Furthermore, \(\ESD\) is invariant to rotations: 
for a rotation matrix \(U\), \(\fSD(UX) = \fSD(X)\).

The gradient of the defining function is given by
\begin{equation}\label{eq:symmetricdirichletdefininggrad}
	\nabla\fSD(X) = X - X^{-\transp}X^{-1}X^{-\transp}
	\;\textrm{.}	
\end{equation}
Unlike \(\fSD\) itself, its gradient
is not symmetric with respect to inverting its argument.
The gradient also features a stronger singularity (\(\sim 1/x^3\)) than the
defining function (\(\sim 1/x\)).

Our optimization method reproduces the output of previous methods
\cite{Rabinovich2017,Claici2017,Liu2018} when optimizing \(\ESD\).
\figref{symmetricdirichletinoursandprevious} shows our method as well as
previous methods used to compute UV maps by minimizing \(\ESD\): the results
visually match.
Unlike these previous methods, we will be able to prove under
which conditions our approach converges.

\subsection{Non-Injective Distortion Energies}
\label{app:otherdeformationenergies}

Beyond flip-free energies like the symmetric Dirichlet energy, many important
energies allow elements to invert.
Although their optima might not be desirable as final results, they have a
large advantage over flip-free energies: they are usually much easier to
optimize, thanks to linearity or a lack of singularities.
Since our algorithm is resilient to initial iterates that contain flips, we can
use optima of these simpler energies as initializers (with exceptions; see
\secref{limitations}).

\subsubsection{Tutte's Energy}
\label{app:tutte}

\begin{definition}[Tutte's energy \cite{Tutte1963}]
    Tutte's energy for the target mesh \(\bW\) is given by
    \begin{equation*}
        \ETutte(\bW) \coloneqq \sum_{\textrm{edges }(i,j)}
        \frac{1}{\norm{\bV_i - \bV_j}} \norm{\bW_i - \bW_j}^2
        \;\textrm{,}
    \end{equation*}
\end{definition}
While Tutte did not define his energy exactly as written above, that definition
can be found in recent references \cite{Jacobson2020}.

Minimizers of \(\ETutte\) for triangulated surfaces are  flip-free if all
boundary vertices are constrained to a convex shape \cite{Tutte1963}.
\(\ETutte\) is a quadratic energy which can be efficiently
minimized by solving a linear system without initialization.
Hence, minimizing \(\ETutte\) is a popular strategy for generating
initial flip-free parametrizations to launch additional line-search-based
optimization steps that further reduce distortion
\cite{Smith2015,Claici2017,Rabinovich2017,Liu2018}.

Minimizers of \(\ETutte\) for \emph{tetrahedralized volumes}, however,  are not
in general flip-free, even if all boundary vertices are constrained to a convex
shape.
This is an obstacle for optimization methods that need to start with a flip-free
map.
Our method
does not automatically fail if the initial map contains flipped elements, and
can thus use \(\ETutte\) even for volumes.

\subsubsection{Conformal Energy}
\label{app:conformal}

As a contrast to Tutte's energy, one can construct quadratic energies built on
estimates of the derivative of a surface/volume map, sensitive to conformal
(angle-based) geometry.
One popular choice is the conformal energy.
\begin{definition}[Conformal energy \cite{Mullen2008}]
    The conformal energy for the target mesh \(\bW\) is
\begin{equation*}
        \EConformal(\bW) \coloneqq
        \frac{1}{2} \sum_{i=1}^m w_i \norm{(G \bW)_i}^2
        - A(\bW)
        \;\textrm{,}
    \end{equation*}
    where \(A(\bW)\) is the area of the target mesh \(\bW\).
\end{definition}
\noindent Similarly to \(\ETutte\), \(\EConformal\) is quadratic in \(\bW\).
Unlike \(\ETutte\), however, minimizers of \(\EConformal\) are not guaranteed to
be flip-free for surfaces.

Several papers propose ways to discretize and optimize \(\EConformal\) in
practice (see \secref{relatedsurfaceparametrization}).
In this work, we employ the method \cite{Sawhney2017}, which efficiently
minimizes \(\EConformal\) with free boundary and minimal area
distortion.

\subsubsection{As-Rigid-As-Possible Energy}
\label{app:arapse}

The linear energies above do not directly measure the deviation of a map from
being \emph{rigid}
The as-rigid-as-possible energy is specifically designed to be sensitive to
non-rigidity.

\begin{definition}[As-Rigid-As-Possible (ARAP) Energy
\cite{Sorkine2007,Liu2008}]
    The ARAP energy's defining function is
	\begin{equation*}
        \fARAP(X) \coloneqq \frac12 \norm{X - \rot X}^2
		\;\textrm{,}
	\end{equation*}
	where \(\rot X\) isolates the rotational part of a matrix \(X\)
	by solving a Procrustes problem \cite{Gower2004},
	\begin{equation*}
		\rot(X)	\coloneqq \argmin_{R \in \SOn{d}} \norm{R - X}^2\;\textrm{.}
	\end{equation*}
    The ARAP energy \(\EARAP\) is defined as the generic
    distortion energy with flips \equref{genericdeformationenergywithflips} and
    \(f=\fARAP\).
\end{definition}

In this article we employ the local-global solver with per-element
discretization \cite{Liu2008} as implemented by libigl \cite{libigl}, which
we denote by \(\EARAP\).
We run the optimization until the relative error between two subsequent iterates
is less than \(10^{-6}\), but not more than 150 minutes.

\(\EARAP\) is a popular distortion energy:
it produces results that are
reminiscent of elasticity, while being cheap to optimize.
Its minimizers are, however, not always flip-free.

\section{Additional Calculations for the Convergence Proof}
\label{app:appendixproof}

This appendix contains proofs for the convergence analysis in
\secref{convergence}.

\begin{proof}[Proof of Lemma \ref{lm:lipc}]
Let $\bP_i = \mathbf{V}\mathbf{\Sigma}{\mathbf{V}}^\transp$ be the eigenvalue
		decomposition of \(\bP_i\).
		We consider \(\ESG\) first (\(f = \fSG\)). 
		Recall that $f(\bP_i) = \frac{1}{2}\|\bP_i\|^2 - \log \det \bP_i$.
		Hence,
			\[
			\|\nabla f(\bP_i)\|^2= \|\bP_i-{\bP_i}^{-1}\|^2 = \|\mathbf{\Sigma}-{\mathbf{\Sigma}}^{-1}\|^2 = \sum_{i=1}^d (\lambda_i-\lambda_i^{-1})^2,
			\]
			where $\mathbf{\Sigma} = \text{Diag}([\lambda_1,\cdots,\lambda_d])$ and $\lambda_1\leq\lambda_2\leq \cdots\leq \lambda_d$.
			By \condref{gradf}, we have 
			\begin{equation}\label{eq:lmbound}
			\lambda_1(\bP_i) \leq C_i^L \triangleq -\frac{B_i}{2} +\frac{\sqrt{4+B_i^2}}{2}.
			\end{equation}
		Next, we aim to use the term $\|\bP_i^{(k+1)}-\bP_i^{(k)}\|$ to bound $\| \nabla f(\bP_i^{(k+1)}) -\nabla f(\bP_i^{(k)})\|$, using the conditions $\lambda_1(\bP_i^{(k+1)}) \leq C_i^L$ and $\lambda_1(\bP_i^{(k)}) \leq C_i^L$, i.e., \eqref{eq:lmbound}.
\begin{align*}
		\|\nabla f(\bP_i^{(k+1)})-\nabla f(\bP_i^{(k)})\| & = \|(\bP_i^{(k+1)}-{\bP_i^{(k+1)}}^{-1}) - (\bP_i^{(k)}-{\bP_i^{(k)}}^{-1})\| \\
		& \leq \|\bP_i^{(k+1)}-\bP_i^{(k)}\|  + \|{\bP_i^{(k+1)}}^{-1}-{\bP_i^{(k)}}^{-1}\|.  
		\end{align*}
We proceed to bound $\|{\bP_i^{(k+1)}}^{-1}-{\bP_i^{(k)}}^{-1}\|$,
		\begin{align*}
	   \|{\bP_i^{(k+1)}}^{-1}-{\bP_i^{(k)}}^{-1}\| & \leq \sqrt{d} \|{\bP_i^{(k+1)}}^{-1}-{\bP_i^{(k)}}^{-1}\|_2 
		= \sqrt{d}\|{\bP_i^{(k)}}^{-1}(\bP_i^{(k)}-\bP_i^{(k+1)}){\bP_i^{(k+1)}}^{-1}\|_2  \\ 
		& \leq \sqrt{d}  \|{\bP_i^{(k)}}^{-1}\|_2\|\bP_i^{(k)}-\bP_i^{(k+1)}\|_2 \|{\bP_i^{(k+1)}}^{-1}\|_2 
		\leq \frac{\sqrt{d}}{{C_i^L}^{2}} \|\bP_i^{(k)}-\bP_i^{(k+1)}\|_2,
		\end{align*}
		where the last inequality follows from spectral norm of inverse coming from first eigenvalue.
		
		Combining the above two inequalities, we find that
		\[ \| \nabla f(\bP_i^{(k+1)}) -\nabla f(\bP_i^{(k)})\| \leq \left( 1+ \frac{\sqrt{d}}{{C_i^L}^{2}} \right) \|\bP_i^{(k+1)}-\bP_i^{(k)}\|.\]

		Following a similar argument, we can also derive the explicit local Lipschitz constant for
		\(\ESD\) (\(f = \fSD\)).
		Recall that $f(\bP_i) = \frac{1}{2}\left(\|\bP_i\|^2 + \|\bP_i^{-1}\|^2\right)$.  Hence,
		\[
		\|\nabla f(\bP_i)\|^2= \|\bP_i-{\bP_i}^{-3}\|^2 = \|\mathbf{\Sigma}-{\mathbf{\Sigma}}^{-3}\|^2 = \sum_{i=1}^d (\lambda_i-\lambda_i^{-3})^2.
		\]
		By \condref{gradf} we have
		$\lambda_1(\bP_i) \leq C_i^{LG}$, i.e., \eqref{eq:lmbound},
		where the $C_i^{LG}$ are the positive roots of the quartic equation
		$x^4+B_ix^3-1=0$.
		Moreover, $C_i^{L} \leq C_i^{LG} \leq 1$ (see \figref{clclgbound}).
		\begin{figure}
			\centering
			\includegraphics[width=0.5\linewidth]{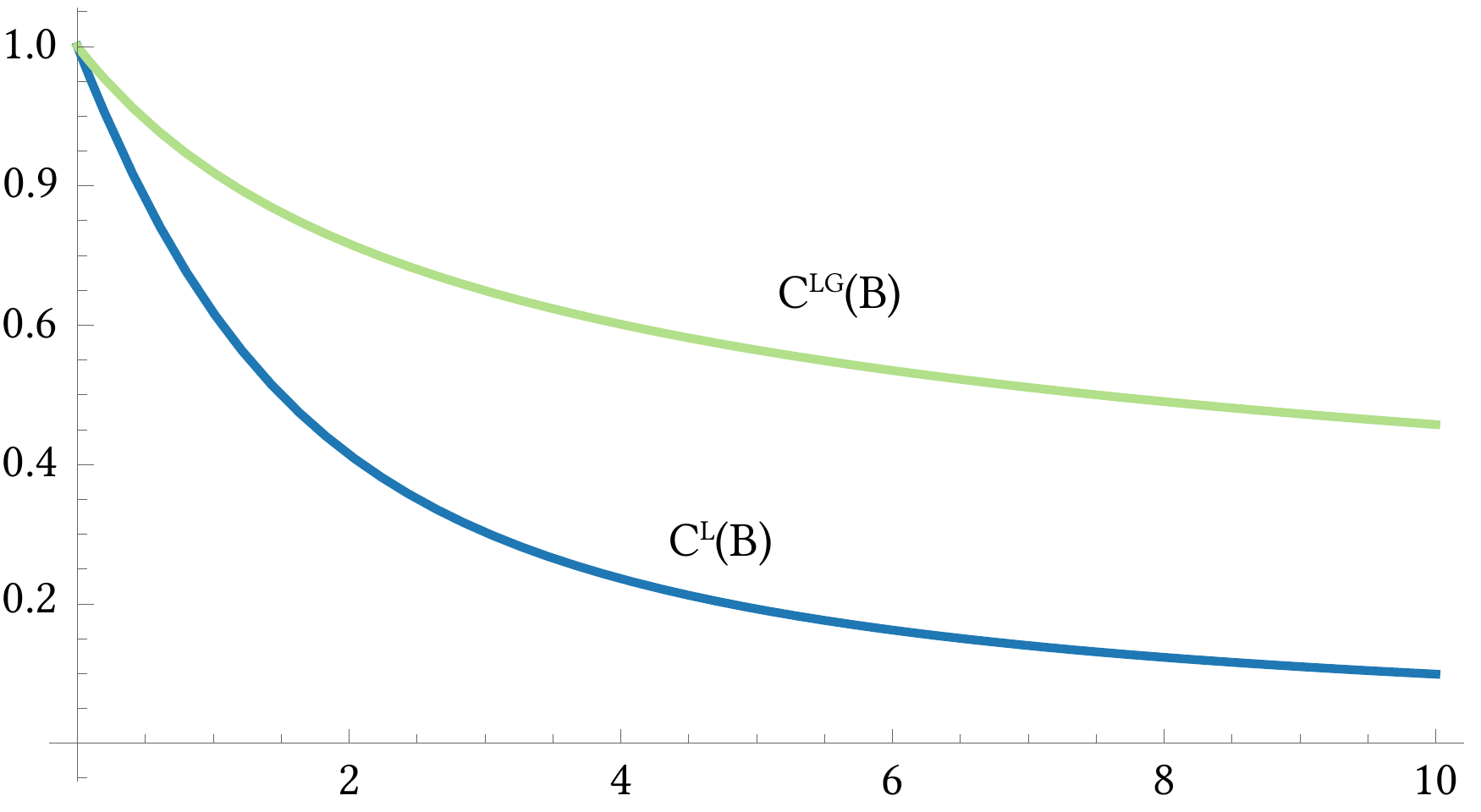}
			\caption{A plot of the bounds \(C_i^{L}\) and \(C_i^{LG}\) with respect
				to the maximal gradient norm \(B_i\).
				One can see that \(C_i^{L} \leq C_i^{LG} \leq 1\).
				\label{fig:clclgbound}}
		\end{figure}
\begin{equation*}
		\begin{aligned}
		\|\nabla f(\bP_i^{(k+1)})-\nabla f(\bP_i^{(k)})\| & = \|(\bP_i^{(k+1)}-{\bP_i^{(k+1)}}^{-3}) - (\bP_i^{(k)}-{\bP_i^{(k)}}^{-3})\| \\
		& \leq \|\bP_i^{(k+1)}-\bP_i^{(k)}\|  + \|{\bP_i^{(k+1)}}^{-3}-{\bP_i^{(k)}}^{-3}\|.  
		\end{aligned}
		\end{equation*}
		Similarly to our proof for the symmetric gradient energy, we next bound $ \|{\bP_i^{(k+1)}}^{-3}-{\bP_i^{(k)}}^{-3}\|$:
		\begin{equation*}
		\begin{aligned}
		\|{\bP_i^{(k+1)}}^{-3}-{\bP_i^{(k)}}^{-3}\|& \leq \sqrt{d}  \|{\bP_i^{(k+1)}}^{-3}-{\bP_i^{(k)}}^{-3}\|_2 \\ 
		& \leq \sqrt{d} \left \|{\bP_i^{(k+1)}}^{-3}\left({\bP_i^{(k+1)}}^3-{\bP_i^{(k)}}^3\right){\bP_i^{(k)}}^{-3}\right\|_2.
		\end{aligned}
		\end{equation*}
		We observe that 
		\begin{equation*}
		\begin{aligned}
		&{\bP_i^{(k+1)}}^3-{\bP_i^{(k)}}^3 \\ &={\bP_i^{(k+1)}}^{2}(\bP_i^{(k+1)}-\bP_i^{(k)})+\bP_i^{(k+1)}(\bP_i^{(k+1)}-\bP_i^{(k)})\bP_i^{(k)}+(\bP_i^{(k+1)}-\bP_i^{(k)}){\bP_i^{(k)}}^2.
		\end{aligned}
		\end{equation*}
		Based on the above equality, we find that
		\begin{align*}
		& \left\|{\bP_i^{(k+1)}}^{-3}\left({\bP_i^{(k+1)}}^3-{\bP_i^{(k)}}^3\right){\bP_i^{(k)}}^{-3}\right\|_2 \\
		 &\quad\leq  \left\|{\bP_i^{(k+1)}}^{-1}\left(\bP_i^{(k+1)}-\bP_i^{(k)}\right){\bP_i^{(k)}}^{-3}\right\|_2
		 + \left\|{\bP_i^{(k+1)}}^{-2}\left(\bP_i^{(k+1)}-\bP_i^{(k)}\right){\bP_i^{(k)}}^{-2}\right\|_2 \\
		 &\quad\quad+ \left\|{\bP_i^{(k+1)}}^{-3}\left(\bP_i^{(k+1)}-\bP_i^{(k)}\right){\bP_i^{(k)}}^{-1}\right\|_2 \\
		 &\quad \leq \frac{3}{{C_i^{LG}}^4}\left\|\bP_i^{(k+1)}-\bP_i^{(k)}\right\|. 
		\end{align*}
		Combining all the statements derived above, we conclude 
		\[
		\|\nabla f(\bP_i^{(k+1)})-\nabla f(\bP_i^{(k)})\| \leq \left( 1 + \frac{3\sqrt{d}}{{C_i^{LG}}^4}\right) \|\bP_i^{(k+1)}-\bP_i^{(k)}\|,
		\]
		which proves the lemma.
\end{proof}

\begin{proof}[Proof of Proposition \ref{prop:admm}]
We begin by deriving a sufficient decrease property for the augmented Lagrangian
function.
The core strategy here is to use the primal blocks $(\bW,\bU,\bP)$ to bound the
dual variable $\bLambda$. 	
\begin{align*}	
\Phi^{k+1}-\Phi^k = & \underbrace{\Phi(\bW^{(k+1)}, \bU^{(k+1)}, \bP^{(k+1)}, \bLambda^{(k+1)})- \Phi(\bW^{(k+1)}, \bU^{(k+1)}, \bP^{(k+1)}, \bLambda^{(k)})}_{(a)} +\\
&\underbrace{\Phi(\bW^{(k+1)}, \bU^{(k+1)}, \bP^{(k+1)}, \bLambda^{(k)})- \Phi(\bW^{(k+1)}, \bU^{(k+1)}, \bP^{(k)}, \bLambda^{(k)})}_{(b)}+ \\
&\underbrace{\Phi(\bW^{(k+1)}, \bU^{(k+1)}, \bP^{(k)}, \bLambda^{(k)})- \Phi(\bW^{(k+1)}, \bU^{(k)}, \bP^{(k)}, \bLambda^{(k)})}_{(c)}+\\ 
&\underbrace{\Phi(\bW^{(k+1)}, \bU^{(k)}, \bP^{(k)}, \bLambda^{(k)})- \Phi(\bW^{(k)}, \bU^{(k)}, \bP^{(k)}, \bLambda^{(k)})}_{(d)}.
\end{align*}
Focusing on the dual update,
\begin{align*}
(a) = & \sum_{i=1}^m \frac{\mu_i}{2} \left(\norm{(G\bW^{(k+1)})_i - \bU_i^{(k+1)} \bP_i^{(k+1)} + \bLambda_i^{(k+1)}}^2
- \norm{\bLambda_i^{(k+1)}}^2 \right) \\
&\quad - \sum_{i=1}^m \frac{\mu_i}{2} \left(\norm{(G\bW^{(k+1)})_i - \bU_i^{(k+1)} \bP_i^{(k+1)} + \bLambda_i^{(k)}}^2
- \norm{\bLambda_i^{(k)}}^2 \right) \!\! \\
= &\sum_{i=1}^m \frac{\mu_i}{2}  \left(\norm{\bLambda_i^{(k+1)} - \bLambda_i^{(k)} + \bLambda_i^{(k+1)}}^2
- \norm{\bLambda_i^{(k+1)}}^2 \right) \\
&\quad- \sum_{i=1}^m \frac{\mu_i}{2} \left(\norm{\bLambda_i^{(k+1)} - \bLambda_i^{(k)} + \bLambda_i^{(k)}}^2
- \norm{\bLambda_i^{(k)}}^2 \right) \!\!\\
= &\sum_{i=1}^m \mu_i \|\bLambda_i^{(k+1)}-\bLambda_i^{(k)}\|^2. 
\end{align*}
To use the primal blocks to bound the dual update $(a)$, we first write down the optimality condition with respect to $\bP^{(k+1)}$, 
\begin{align*}
&  \bP^{(k+1)} =  \mathop{\argmin_{\bP \in (\SPDn{d})^m}} 
\Phi\left(\bW^{(k+1)}, \bU^{(k+1)}, \bP, \bLambda^{(k)}\right)\\
\Rightarrow \, & 0 = w_i \nabla f(\bP_i^{(k+1)}) + \mu_i \operatorname{symm}\left({\bU_i^{(k+1)}}^\transp\left(\bU_i^{(k+1)}\bP_i^{(k+1)}-(G\bW^{(k+1)})_i-\bLambda_i^{(k)}\right)\right), \;\; \forall i \\
\Rightarrow \, & 0 = w_i \nabla f(\bP_i^{(k+1)}) - \mu_i  \operatorname{symm}\left({\bU_i^{(k+1)}}^\transp\bLambda_i^{(k+1)}\right), \;\; \forall i.
\end{align*}
Then, we have 
\begin{align*}
& \frac{w_i}{\mu_i}\bU_i^{(k+1)} \nabla f(\bP_i^{(k+1)})  = \frac{1}{2}\left(\bLambda_i^{(k+1)} +\bU_i^{(k+1)} {\bLambda_i^{(k+1)}}^\transp \bU_i^{(k+1)}\right) ,\forall i.
\end{align*}
Thus, we can use $(\bU,\bP)$ to bound $(a)$. By \condref{lambdabound},
	\begin{align*}
	& \left\|\bLambda_i^{(k+1)}-\bLambda_i^{(k)}\right\|^2 \leq  \gamma \left\|\frac{1}{2}\left(\bLambda_i^{(k+1)}-\bLambda_i^{(k)}\right) +  \frac{1}{2}\left(\bU_i^{(k+1)} {\bLambda_i^{(k+1)}}^\transp \bU_i^{(k+1)}-\bU_i^{(k)} {\bLambda_i^{(k)}}^\transp \bU_i^{(k)}\right)\right\|^2.
	\end{align*}
Thus,
\begin{align*}
&\|\bLambda_i^{(k+1)}-\bLambda_i^{(k)}\|^2 \\  \leq \, & 
\frac{\gamma w_i^2}{\mu_i^2}\left \|\bU_i^{(k+1)}\nabla f(\bP_i^{(k+1)})-\bU_i^{(k)}\nabla f(\bP_i^{(k)})\right\|^2 \\
= \, & \frac{\gamma w_i^2}{\mu_i^2} \|\bU_i^{(k+1)}\nabla f(\bP_i^{(k+1)})-\bU_i^{(k)}\nabla f(\bP_i^{(k+1)})  +\bU_i^{(k)}\nabla f(\bP_i^{(k+1)})-\bU_i^{(k)}\nabla f(\bP_i^{(k)})\|^2 \\
\leq\, & \frac{2\gamma w_i^2}{\mu_i^2}  \left(B_i^2\|\bU_i^{(k+1)}-\bU_i^{(k)}\|^2 +\|\nabla f(\bP_i^{(k+1)})-\nabla f(\bP_i^{(k)})\|^2\right) \\ 
\leq \, & \frac{2\gamma w_i^2B_i^2}{\mu_i^2}  \|\bU_i^{(k+1)}-\bU_i^{(k)}\|^2 +\frac{2\gamma w_i^2F_i^2}{\mu_i^2} \|\bP_i^{(k+1)}-\bP_i^{(k)}\|^2,
\end{align*}
where the last inequality follows from lemma \ref{lm:lipc}. 
Based on the above analysis,
\begin{align*}
	(a) & = \sum_{i=1}^m \mu_i \|\bLambda_i^{(k+1)}-\bLambda_i^{(k)}\|^2 \\
	& \leq  \sum_{i=1}^m \frac{2\gamma w_i^2B_i^2}{\mu_i}  \|\bU_i^{(k+1)}-\bU_i^{(k)}\|^2 +\frac{2\gamma w_i^2F_i^2}{\mu_i} \|\bP_i^{(k+1)}-\bP_i^{(k)}\|^2.
\end{align*}

We continue by bounding the terms $(b)$, $(c)$, $(d)$.
As
\(\Phi(\bW^{(k+1)},$ $ \bU^{(k+1)}, \bP, \bLambda^{(k)})\) is $(w_i+\mu_i)$-strongly convex for $\bP_i$, we have, 
\[
(b) \leq -\sum_{i=1}^m\frac{w_i+\mu_i}{2} \|\bP_i^{(k+1)}-\bP_i^{(k)}\|^2.
\]
See \cite[Theorem 2.1.8]{nesterov2018lectures} for details.  
Since $\bU^{(k+1)}$ is the global optimal solution of \linebreak
\(\Phi(\bW^{(k+1)}, \bU, \bP^{(k)}, \bLambda^{(k)}) + \sum_{i=1}^m \frac{h_i}{2}\|\bU_i - \bU_i^{(k)}\|^2\), we obtain 
\[
(c)\leq -\sum_{i=1}^m\frac{h_i}{2} \|\bU_i^{(k+1)}-\bU_i^{(k)}\|^2.
\] 
Similarly, we can derive the associated sufficient decrease term for $\bW$, i.e., 
\[
(d)\leq -\frac{1}{2}\lambda_{\text{min}}(L)\|\bW^{(k+1)}-\bW^{(k)}\|^2.
\] 
By summing up all the inequalities for $(a),(b),(c),(d)$,   
\begin{align*}
\Phi^{k+1}-\Phi^k \leq &- \sum_{i=1}^m   \left(\frac{h_i}{2} -\frac{2\gamma w_i^2B_i^2}{\mu_i}\right)\|\bU^{(k+1)}_i-\bU^{(k)}_i\|^2 \\
& -\sum_{i=1}^m\left(\frac{w_i+\mu_i}{2} - \frac{2\gamma w_i^2F_i^2}{\mu_i}\right) \|\bP^{(k+1)}_i-\bP^{(k)}_i\|^2 \\
& -\frac{1}{2}\lambda_{\text{min}}(L)\|\bW^{(k+1)}-\bW^{(k)}\|^2.
\end{align*}

We now apply the conditions
\(\mu_i>  -\frac{1}{2}(w_i-2\epsilon)$ $ +\frac{1}{2} \sqrt{(w_i-2\epsilon)^2+16\gamma w_i^2F_i^2}\)
and \linebreak \(h_i \ge \frac{4\gamma w_i^2B_i^2}{{\mu_i}} + 2 \epsilon\)
to arrive at 
 \begin{equation}\label{eq:suff}
\begin{aligned}
\Phi^{k+1}-\Phi^k \leq& -\frac{1}{2}\lambda_{\text{min}}(L)\|
\bW^{(k+1)}-\bW^{(k)}\|^2\\
&-\sum_{t=1}^m \epsilon\Big(\|\bU_i^{(k+1)}-\bU_i^{(k)}\|^2 +\|\bP_i^{(k+1)}-\bP_i^{(k)}\|^2\Big),
\end{aligned}
\end{equation}
which proves the statement of the theorem.

\end{proof}

\begin{proof}[Proof of Theorem \ref{main-theorem}] 
There are four core steps to complete this proof.
\begin{itemize}
\item \textbf{Step 1:} Show that the sequence $\{(\bW^{(k)}, \bU^{(k)}, \bP^{(k)}, \bLambda^{(k)})\}_{k=0}^\infty$ is bounded. 
\end{itemize}
The boundedness of the sequence  $\{\bP^{(k)}\}_{k=0}^\infty$ follows directly
from \condref{gradf}.
$\bU_i$ is a rotation matrix and thus bounded.
Recall that 
\begin{equation}\label{eq:primal-dual}
	\begin{aligned}
	\left\|\bLambda_i^{(k+1)}-\bLambda_i^{(k)}\right\|^2 \leq \frac{2\gamma w_i^2B_i^2}{\mu_i^2}  \|\bU_i^{(k+1)}-\bU_i^{(k)}\|^2 +\frac{2\gamma w_i^2F_i^2}{\mu_i^2} \|\bP_i^{(k+1)}-\bP_i^{(k)}\|^2.
	\end{aligned}
\end{equation}
Therefore, we can conclude that the dual variable $\bLambda$ is bounded. Using the update rule for $\bW$ directly gives a bound for $\bW$.
Hence, the sequence $\{(\bW^{(k)}, \bU^{(k)}, \bP^{(k)}, \bLambda^{(k)})\}_{k=0}^\infty$ is bounded, and thus a cluster point exists. 
\begin{itemize}
	\item \textbf{Step 2:} Prove that
	$\lim\limits_{k\rightarrow +\infty} \|\bU^{(k+1)}-\bU^{(k)}\|^2+ \|\bP^{(k+1)}-\bP^{(k)}\|^2+ \|\bW^{(k+1)}-\bW^{(k)}\|^2 +  \|\bLambda^{(k+1)}-\bLambda^{(k)}\|^2 =0$,
	where the squared norm of \(\mathbf{A}\) indicates the appropriate sum over
	all the squared norms of the \(\mathbf{A}_i\).
\end{itemize}
Suppose that $(\bW^{*}, \bU^{*}, \bP^{*}, \bLambda^{*})$ is a cluster point of the sequence $\{(\bW^{(k)}, \bU^{(k)}, \bP^{(k)}, \bLambda^{(k)})\}_{k=0}^\infty$. Let  $\{(\bW^{(k_i)}, \bU^{(k_i)}, \bP^{(k_i)}, \bLambda^{(k_i)})\}$ be a convergent subsequence such that 
\[
\lim_{i \rightarrow +\infty }(\bW^{(k_i)}, \bU^{(k_i)}, \bP^{(k_i)}, \bLambda^{(k_i)})= (\bW^{*}, \bU^{*}, \bP^{*}, \bLambda^{*}).
\]
By summing \equref{suff} from $k=0$ to $k = k_i-1$, we have 
\begin{align*}
&\Phi(\bW^{(k_i)}, \bU^{(k_i)}, \bP^{(k_i)}, \bLambda^{(k_i)})- \Phi(\bW^{(0)}, \bU^{(0)}, \bP^{(0)}, \bLambda^{(0)}) \\
&\quad \leq -\frac{1}{2}\lambda_{\text{min}}(L) \sum_{k=0}^{k_i-1}\|
\bW^{(k+1)}-\bW^{(k)}\|^2
 -\sum_{k=0}^{k_i-1}\sum_{t=1}^m \epsilon\left(\|\bU_i^{(k+1)}-\bU_i^{(k)}\|^2+\|\bP_i^{(k+1)}-\bP_i^{(k)}\|^2\right). 
\end{align*}
Taking the limit of $i \rightarrow +\infty$ in above inequality and rearranging terms, we obtain
\begin{equation}\begin{split}\label{eq:lcbimplies}
& \frac{1}{2}\lambda_{\text{min}}(L) \sum_{k=0}^{+\infty}\|
\bW^{(k+1)}-\bW^{(k)}\|^2
+ \sum_{k=0}^{+\infty}\sum_{t=1}^m \epsilon\left(\|\bU_i^{(k+1)}-\bU_i^{(k)}\|^2+\|\bP_i^{(k+1)}-\bP_i^{(k)}\|^2\right) \\
& \quad \leq \Phi(\bW^{(0)}, \bU^{(0)}, \bP^{(0)}, \bLambda^{(0)}) - \Phi(\bW^{*}, \bU^{*}, \bP^{*}, \bLambda^{*}) < \infty.
\end{split}\end{equation}
Here the last inequality holds as our augmented Lagrangian function is unbounded only if our input is unbounded. Moreover, the sequence $\{(\bW^{(k)}, \bU^{(k)}, \bP^{(k)}, \bLambda^{(k)})\}_{k=0}^\infty$ is bounded and we can complete the argument. 

\equref{lcbimplies} implies that  
\[
\sum_{k=0}^{+\infty}\|
\bW^{(k+1)}-\bW^{(k)}\|^2 < \infty , \quad
\sum_{k=0}^{+\infty}\sum_{i=1}^m\|
\bU_i^{(k+1)}-\bU_i^{(k)}\|^2 < \infty , \quad
\sum_{k=0}^{+\infty}\sum_{i=1}^m\|
\bP_i^{(k+1)}-\bP_i^{(k)}\|^2  < \infty.
\]
Hence,
\(\bW^{(k+1)}-\bW^{(k)} \rightarrow 0,\bU^{(k+1)}-\bU^{(k)} \rightarrow 0 ,\bP^{(k+1)}-\bP^{(k)} \rightarrow 0. \)
Due to the primal-dual relationship \eqref{eq:primal-dual}, we can thus conclude that $\bLambda^{(k+1)}-\bLambda^{(k)} \rightarrow 0$.

\begin{itemize}
	\item \textbf{Step 3}: Derive a safeguard property.
\end{itemize}

Define the extended augmented Lagrangian function
\(G(\bW, \bU, \bP,$ $ \bLambda) = \Phi(\bW, \bU, \bP, \bLambda) + \sum_{i=1}^m g(\bU_i)\).
Recall the optimization optimality conditions for the ADMM updates (i.e., the $k+1$ iteration). 
\begin{equation*}
\left\{
\begin{aligned}
& 0 =  h_i\left(\bU_i^{(k+1)}-\bU_i^{(k)}\right) + \mu_i\left(\bU_i^{(k+1)}\bP_i^{(k)}-(G\bW^{(k+1)})_i-\bLambda_i^{(k)}\right){\bP_i^{(k)}}^\transp
+\partial g(\bU_i^{(k+1)}), \forall i  \\
& 0 =  w_i \nabla f(\bP_i^{(k+1)}) - \mu_i \operatorname{symm}\left({\bU_i^{(k+1)}}^\transp\bLambda_i^{(k+1)}\right),\forall i\\
& \bLambda_i^{(k+1)} = \bLambda_i^{(k)} + (G\bW^{(k+1)})_i-\bU_i^{(k+1)}\bP_i^{(k+1)}.
\end{aligned}
\right.
\end{equation*}
Moreover, a stationary point satisfying
\(0 \in \partial G(\bW^*, \bU^*, \bP^*, \bLambda^*)\) is equivalent to the KKT
point property in \equref{kkt}.
Subsequently, we want to bound the subgradient \linebreak
\(\text{dist}(0, \partial G(\bW^{(k+1)}, \bU^{(k+1)}, \bP^{(k+1)}, \bLambda^{(k+1)}))\)
by the iterate difference, i.e.
\(\|\bP^{(k+1)}-\bP^{(k)}\|\), \(\|\bU^{(k+1)}-\bU^{(k)}\|\),
\(\|\bW^{(k+1)}-\bW^{(k)}\|\),
\begin{align*}
& \text{dist}\left(0, \partial G(\bW^{(k+1)}, \bU^{(k+1)}, \bP^{(k+1)}, \bLambda^{(k+1)}) \right) \leq \\
&\sum_{i=1}^m \left \|w_i \nabla f(\bP_i^{(k+1)}) -\mu_i\operatorname{symm}\left({\bU_i^{(k+1)}}^\transp\bLambda_i^{(k+1)}\right)\right \| \\
& + \sum_{i=1}^m \text{dist}\left(0,\partial g({\bU_i^{(k+1)}})- \mu_i\bLambda_i^{(k+1)}{\bP_i^{(k+1)}}^\transp\right)
 + \sum_{i=1}^m \left \|(G\bW^{(k+1)})_i - \bU_i^{(k+1)}\bP_i^{(k+1)}\right\|.
\end{align*}

We observe that the first term is 0, and the third term is identical to $\sum_{i=1}^m \|\bLambda_i^{(k+1)} - \bLambda_i^{(k)}\|$.
It remains to bound the second term.
Starting from the optimality condition w.r.t $\bU_i^{(k+1)}$,
\begin{align*}
0 = & h_i\left(\bU_i^{(k+1)}-\bU_i^{(k)}\right) + \mu_i\left(\bU_i^{(k+1)}\bP_i^{(k)}-(G\bW^{(k+1)})_i-\bLambda_i^{(k)}\right){\bP_i^{(k)}}^\transp
 +\partial g(\bU_i^{(k+1)})\\
0 = & h_i\left(\bU_i^{(k+1)}-\bU_i^{(k)}\right) + \mu_i\left (\bLambda_i^{(k+1)}+ \bU_i^{(k+1)}\left(\bP_i^{(k)}-\bP_i^{(k+1)}\right)\right){\bP_i^{(k)}}^\transp 
+\partial g(\bU_i^{(k+1)})\\
0 = &  h_i\left(\bU_i^{(k+1)}-\bU_i^{(k)}\right) + \mu_i\bLambda_i^{(k+1)}\left({\bP_i^{(k)}}^\transp-{\bP_i^{(k+1)}}^\transp\right)+\\
& \mu_i \bU_i^{(k+1)}\left(\bP_i^{(k)}-\bP_i^{(k+1)}\right){\bP_i^{(k)}}^\transp +\mu_i\bLambda_i^{(k+1)}{\bP_i^{(k+1)}}^\transp+\partial g(\bU_i^{(k+1)}) .
\end{align*}
As $\{\bW^{(k+1)}, \bU^{(k+1)}, \bP^{(k+1)}, \bLambda^{(k+1)}\}_{k\ge 0}$ is bounded, i.e., \textbf{step 1},
there exists a constant $D$ such that 
\begin{align*}
&(\bW^{(k+1)}, \bU^{(k+1)}, \bP^{(k+1)},\bLambda^{(k+1)}) \in \mathcal{C},\,\, \mathcal{C} = \left\{(\bW, \bU, \bP, \bLambda)| \|\bW, \bU, \bP, \bLambda\|\leq D\right\}.
\end{align*}
Thus,
\begin{align*}
	& \text{dist}\left(0,\mu_i\bLambda_i^{(k+1)}{\bP_i^{(k+1)}}^\transp+\partial g(\bU_i^{(k+1)})\right) \\
	& \leq h_i \|\bU_i^{(k+1)}-\bU_i^{(k)}\| + \mu_i\left(D+\|\bP_i^{(k)}\|\right)\|{\bP_i^{(k)}}-{\bP_i^{(k+1)}}\| .
\end{align*}
Due to Lemma \ref{lm:lipc}, we know that there exists a constant $\kappa>0$ 
such that
\begin{align*}
& \text{dist}\left(0, \partial G(\bW^{(k+1)}, \bU^{(k+1)}, \bP^{(k+1)}, \bLambda^{(k+1)}) \right) \\
\leq \, &\kappa  \left(\|\bU^{(k+1)}-\bU^{(k)}\|+\|\bP^{(k+1)}-\bP^{(k)}\|+\|\bLambda^{(k+1)}-\bLambda^{(k)}\|\right).
\end{align*}
Based on Step 2, i.e., $\bW^{(k+1)}-\bW^{(k)} \rightarrow 0,\bU^{(k+1)}-\bU^{(k)} \rightarrow 0 ,\bP^{(k+1)}-\bP^{(k)} \rightarrow 0$, there exists
$d_{k+1} \in \partial G(\bW^{(k+1)}, \bU^{(k+1)}, \bP^{(k+1)}, \bLambda^{(k+1)}) $ such that $\|d_{k+1}\| \rightarrow 0$. By the definition of general subgradient, we have $0\in \partial G(\bW^{*}, \bU^{*}, \bP^{*}, \bLambda^{*})$. Thus, any cluster point $(\bW^{*}, \bU^{*}, \bP^{*}, \bLambda^{*})$ of a sequence $(\bW^{(k)}, \bU^{(k)},$ $ \bP^{(k)}, \bLambda^{(k)})$ generated by the ADMM is a stationary point, or KKT point equivalently.

\begin{itemize}
	\item \textbf{Step 4:} Show that $G(\bW, \bU, \bP, \bLambda)$  is a
	Kurdyka-\L{}ojasiewicz function.
\end{itemize}
Following the proof of Theorem 2.9 in \cite{attouch2013convergence},
we can infer the global convergence of the sequence $\{\bW^{(k)}, \bU^{(k)}, $ $\bP^{(k)}, \bLambda^{(k)}\}$ from the K\L{} condition of the extended augmented Lagrangian function
$G(\bW, \bU, \bP, \bLambda)$.
Therefore, the final step is to prove that $G(\bW, \bU, \bP, \bLambda)$  is a
Kurdyka-\L{}ojasiewicz function. 

Recall that 
\begin{align*}
 G(\bW, \bU, \bP, \bLambda) = & \sum_{i=1}^m w_i f(\bP_i) 
  + \sum_{i=1}^m \frac{\mu_i}{2}\left(\|(G\bW)_i-\bU_i\bP_i + \bLambda_i\|^2 - \|\bLambda_i\|^2 \right)
  + \sum_{i=1}^m g(\bU_i). 
\end{align*}

The K\L{} property is closed under summation \cite{attouch2013convergence}.
Thus, we can check the above summands one by one.
$\sum_{i=1}^m w_i f(\bP_i)$ is strongly convex and hence satisfies the uniform
convexity property, and is a K\L{} function \cite[Section 4.1]{attouch2010proximal}.
$\sum_{i=1}^m \frac{\mu_i}{2}\left(\|(G\bW)_i-\bU_i\bP_i + \bLambda_i\|^2 - \|\bLambda_i\|^2 \right) $ is a polynomial function and thus semi-algebraic, and semi-algebraic functions satisfy the K\L{} property \cite{attouch2010proximal,attouch2013convergence}.
As $g(\cdot)$ is the indicator function over the special orthogonal group, it
is a  K\L{} function (via Stiefel manifolds \cite{attouch2010proximal}).

This completes the proof of the theorem.
\end{proof} \end{appendices}

\clearpage 
\bibliographystyle{plain}
\bibliography{bibliography.bib}

\clearpage
\pagenumbering{arabic}
\setcounter{page}{1}
\setcounter{section}{0}

\section{Supplemental: Implementation Details}

This appendix contains details needed to implement our splitting method.
For this supplemental material,
\(\varepsilon_m\) is the machine epsilon of the chosen
floating point type.

\subsection{Computing the Jacobian Map}

For triangle and tetrahedral meshes we compute the Jacobian of the
map from \(\bV\) to \(\bW\) using the
gradient operator for piecewise linear Langrangian finite elements.
The gradient vector of the \(k\)-th coordinate function of \(\bW\) with
with respect to the source mesh \(\bV\) on the triangle/tetrahedron
\(j\) corresponds to the \(k\)-th column of the Jacobian matrix
on the element \(j\).
For more background on interpreting the Jacobian as a
finite element gradient, see \cite{Pinkall1993}.

On surfaces, where we need to compute a map from
\(\Ro^3\) to \(\Ro^2\), we use the intrinsic gradient matrix,
 to get Jacobians in \(\Rnm{2}{2}\)
\cite[\texttt{grad\textunderscore{}intrinsic.h}]{libigl}.
For volumes, where we are computing a map from
\(\Ro^3\) to \(\Ro^3\), we use the standard coordinate-aligned
gradient matrix \cite[\texttt{grad.h}]{libigl}.

\subsection{Solving the Optimization in \(\bP\)}

To perform the optimization step in \(\bP\), we need to solve
\equref{rootfindinginp}.
This section explains our approach to solving equations of the form
\begin{equation}\label{eq:genericpeqtosolve}
	w \nabla f(P) + \mu P = \mu Q
\end{equation}
for \(P \in \SPDn{d}\), where \(Q\) is a symmetric matrix.

\subsubsection{Symmetric Gradient Energy}

For \(f = \fSG\), \(\nabla f(P) = P - P^{-1}\).
Thus \equref{genericpeqtosolve} becomes
\begin{equation}\label{eq:peqfsgpluggedin}
	(w + \mu) P^2 - \mu Q P - wI = 0
	\;\textrm{,}	
\end{equation}
where \(I\) is the identity matrix.
\equref{peqfsgpluggedin} is a quadratic equation in \(P\) and has a single symmetric positive definite solution,
which can be obtained using the regular quadratic formula:
\begin{equation}\label{eq:solvesgwithquadraticformula}
	P = \frac{1}{2(w + \mu)}
	\left( \mu Q + \sqrt{\mu^2 Q^2 + 4 w (w + \mu)I} \right)
	\;\textrm{.}
\end{equation}

We compute the matrix square root for \(d=2,3\) using \cite{Franca1989}.
If we determine that this method can not be used reliably because
of floating point issues (the discriminant, as of \cite{Franca1989}, is
smaller than \(\sqrt{\varepsilon_m}\)), we perform an eigendecomposition and
compute the square root of all eigenvalues instead.
If we determine \(Q\) to be very small (\(\norm{Q}^2 < \sqrt{\varepsilon_m}\)),
we employ a Taylor approximation of \equref{solvesgwithquadraticformula} in
\(Q\).

\subsubsection{Symmetric Dirichlet Energy}

For \(f = \fSD\), \(\nabla f(P) = P - P^{-3}\).
Thus \equref{genericpeqtosolve} becomes
\begin{equation}\label{eq:peqfsdpluggedin}
	(w + \mu) P^4 - \mu Q P^3 - wI = 0
	\;\textrm{,}	
\end{equation}
where \(I\) is the identity matrix.
\equref{peqfsdpluggedin} is a quartic equation, for which we know there
is
a unique symmetric positive definite solution, as \(f\) is convex.
We solve this quartic equation by applying eigendecomposition,
transforming the problem into \(d\) scalar problems in the eigenvalues,
and using the explicit quartic root finding method \cite{Khashin2020}
to find the unique positive solution to the scalar quartic equation.
If, due to floating point issues, the quartic solver fails to find
a result that is within a specified tolerance, we improve the solver's
result using Newton root finding.
\\

For both energies, if the determinant or trace of \(P\) are smaller than
\(\sqrt{\varepsilon_m}\), we explicitly ensure that its eigenvalues are
at least \(\sqrt{\varepsilon_m}\).

\subsection{Solving the Optimization in \(\bU\)}

To perform the optimization step in \(\bU\), we need to solve
the Procrustes problem \equref{optimizationinudecoupled}.

\subsubsection{\(d=2\)}

In two dimensions, we employ our own simple Procrustes solver.
Our goal is to find
\begin{equation}\begin{split}\label{eq:twodprocrustes}
	\varphi &= \argmin_\varphi \norm{U(\varphi) - Q}^2 =
	\argmin_\varphi \left( -U(\varphi) \cdot Q \right) , \\
	&\textrm{where } U(\varphi) \coloneqq \begin{pmatrix}
 		\cos\varphi & -\sin\varphi \\
 		\sin\varphi & \cos\varphi
 	\end{pmatrix}
	\;\textrm{,}
\end{split}\end{equation}
given an arbitrary \(Q \in \Rnm{2}{2}\).
The objective function from \equref{twodprocrustes} will
attain its minimum at a critical point of its objective function,
which is a root of a simple trigonometric equation that
can be solved using the \(\atant\) function.
The second derivative of the objective function
is then used to pick out the minimum among the critical points.

In practice, we store the rotations in \(\bU\)
as the real and
imaginary parts of a complex number, so that \(\atant\)
does not need to be computed using any trigonometric functions: we can simply
employ a square root.

\subsubsection{\(d=3\)}

In three dimensions, we use a standard implementation
\cite[\texttt{polar\textunderscore{}svd.h}]{libigl},
which computes a singular value decomposition to solve the Procrustes problem.
 
\end{document}